\documentclass[11pt, reqno]{amsart}
\usepackage{graphicx,amsmath,amssymb,epsfig,xcolor}
\definecolor{mygray}{gray}{0.8}
\usepackage{tabularx}
\usepackage{array}
\usepackage{arydshln}
\usepackage{footnote}
\usepackage{caption}
\usepackage{subcaption}
\usepackage{floatrow}
\usepackage{booktabs}
\usepackage{amsthm,thmtools}
\usepackage{amsfonts, lscape,comment}
\usepackage{color}
\usepackage{bbm}
\usepackage{harvard,verbatim}
\usepackage{threeparttable}
\usepackage[noend]{algpseudocode}

\usepackage{url}
\newtheorem{theorem}{Theorem}
\theoremstyle{plain}

\newtheorem{assumption}{Assumption}

\newtheorem{algorithm}{Algorithm}

\newtheorem{corollary}{Corollary}

\newtheorem{lemma}{Lemma}

\newtheorem{proposition}{Proposition}

\theoremstyle{definition}
\newtheorem{definition}{Definition}
\newtheorem{example}{Example}[section]

\newenvironment{varproof}[1][Proof]{\begin{trivlist}
\item[\hskip \labelsep {\bfseries #1}]}{\end{trivlist}}
\numberwithin{equation}{section}
\algrenewcommand\algorithmiccomment[1]{\hfill #1}

\theoremstyle{definition}
\newtheorem{remark}{Remark}[section]
\graphicspath{ {FigsEtc/} }

\begin{document}
\title[\textbf{Yogurts Choose Consumers}]{Yogurts Choose Consumers?
Estimation of Random-Utility Models via Two-Sided Matching}
\author[O. Bonnet]{Odran Bonnet$^{\flat }$}
\author[A. Galichon]{Alfred Galichon$^{\dag }$ }
\author[Y.-W. Hsieh]{Yu-Wei Hsieh$^{\lozenge}$}
\author[K. O'Hara]{Keith O'Hara$^{\lozenge }$}
\author[M. Shum]{Matt Shum{\small $^{\S }$}}
\date{\today\ (First draft: 9/2015). Galichon gratefully acknowledges
funding from a grant NSF DMS-1716489, and from a European Research Council (ERC) grant No. 866274. We thank Jeremy Fox, Xavier d'Haultefoeuille, Lars Nesheim, Ariel Pakes, and
participants in seminars at Amazon, Caltech, UC-Davis, Emory, Johns Hopkins, Indiana, LSE, Rochester (Simon), Stanford, UNC, USC, Yale, NYU CRATE
conference, Banff Applied Microeconomics Conference, SHUFE
Econometrics Conference,  Toronto Intersections of Econometrics and
Applied Micro Conference, WARP (Workshop of Applications of Revealed Preference) webinar, and UCL-Vanderbilt Conference on Econometrics and
Models of Strategic Interactions for helpful comments. Alejandro
Robinson-Cortes provided excellent research assistance.\\
{\indent $^{\flat }$INSEE/CREST, Paris; odran.bonnet@insee.fr }\\
{\indent$^{\dag }$New York University; ag133@nyu.edu }\\
{\indent$^{\lozenge }$Amazon.com; yuweihsieh01@gmail.com and keith.ohara@nyu.edu. Hsieh and O'Hara's contributions to the paper reflect work done prior to their joining Amazon.}\\
{\indent $^{\S }$California Institute of Technology; mshum@caltech.edu} }

\begin{abstract}
The problem of {\em demand inversion} -- a crucial step in the estimation of random utility discrete-choice models -- is equivalent to the determination of stable outcomes in
two-sided matching models. This equivalence applies to random utility
models that are not necessarily additive, smooth, nor even invertible. Based on this
equivalence, algorithms for the determination of stable
matchings provide effective computational methods for estimating these models.  For non-invertible models, the identified set of utility vectors is a lattice, and the matching algorithms recover sharp upper and lower bounds on the utilities.   Our matching approach
facilitates estimation of models that were previously difficult to
estimate, such as the pure characteristics model.
An empirical application to voting data from the 1999 European Parliament elections illustrates the good performance of our matching-based demand inversion algorithms in practice.

\vspace{0.1cm} \emph{Keywords:} random utility models, demand inversion, two-sided matching,
discrete-choice demand models, partial identification, pure characteristics model

\vspace{0.1cm} \emph{JEL Classification:} C51, C60

\vspace{0.5cm}

\end{abstract}

\maketitle




\bigskip

\newpage

\section{\setcounter{page}{2}\setcounter{equation}{0}Introduction}

Discrete choice models play a
tremendous role in applied work in economics. In these models, an agent $i$
characterized by a utility shock $\varepsilon _{i}\in \Omega $ must choose
from a finite set of alternatives $j\in \mathcal{J}$ in order to maximize
her utility. The random utility framework pioneered by~\citeasnoun{mcfadden1978modelling} assumes that the utility $\mathcal{U}_{\varepsilon _{i}j}\left( \delta
_{j}\right) $ that agent $i$ gets from alternative $j$ depends on $\delta _{j}$, a systematic utility level associated with
alternative $j$ (possibly parameterized as a function of regressors such as characteristics of alternative $j$) which is identical across all agents, and a realization $%
\varepsilon _{i}$ of agent $i$'s random utility shocks. 
The agent chooses the alternative yielding maximal utility:
\begin{equation}
\max_{j\in \mathcal{J}}\left\{ \mathcal{U}_{\varepsilon _{i}j}\left( \delta
_{j}\right) \right\} .  \label{consumerProgram}
\end{equation}%
We focus on parametric random utility models, where the function $\mathcal{U}:\left( \varepsilon ,j,\delta
\right) \in \Omega \times \mathcal{J}\times \mathbb{R}\mapsto \mathcal{U}%
_{\varepsilon j}\left( \delta \right) \in \mathbb{R}$ as well as the
distribution of $\varepsilon $ in the population, denoted $P$, are known to
the researcher. Particular instances of these models are Additive Random Utility
Models (hereafter ARUMs), including Logit or Probit models, where $\Omega =\mathbb{R}%
^{\mathcal{J}}$, $\mathcal{U}_{\varepsilon _{i}j}\left( \delta _{j}\right)
=\delta _{j}+\varepsilon _{ij}$, and $\varepsilon $ follows a Gumbel or
Gaussian distribution. However, our results extend to the more general
class of Non-Additive Random Utility Models (NARUMs), in which $\mathcal{U}%
_{\varepsilon _{i}j}\left( \delta _{j}\right) $ is not 
(quasi-)linear in $\delta _{j}$. 

Under an assumption guaranteeing that agents are not indifferent
between any pair of alternatives (see Assumption 2 below), we can define the vector-valued \emph{demand map }$\sigma \left(
.\right) $, the $j$-th component ($j\in\mathcal{J}_0$) of which is defined as\ the probability that alternative $j$ dominates all
the other ones, given the vector of systematic utilities $\left( \delta
_{j}\right) _{j\in \mathcal{J}}$:
\begin{equation}
\sigma _{j}\left( \delta \right) =P\left( \varepsilon :\;\mathcal{U}%
_{\varepsilon j}\left( \delta _{j}\right) \geq \mathcal{U}_{\varepsilon
j^{\prime }}\left( \delta _{j^{\prime }}\right) ,\;\forall j^{\prime }\in
\mathcal{J}\right) .  \label{ProbaDetermination}
\end{equation}
The main focus of the paper pertains to {\em demand inversion}: given a vector of  observed
market shares $\left( s_{j}\right) _{j\in \mathcal{J}}$, how can one
characterize and compute the full set of utility vectors $\left( \delta
_{j}\right) _{j\in \mathcal{J}}$ such that $s=\sigma \left( \delta \right)$
 -- that is, which rationalizes the observed market shares? (In cases where the distribution $P$ of the utility shocks depends on parameters unknown to the researcher, which arises in many applications as well as in our simulations and application below, demand inversion refers to recovering $\delta_j$ for a given set of parameters of $P$.) Additionally, the demand map may also be {\em non-invertible}, which arises when multiple vectors
$\delta $ solve the demand inversion problem.

\subsection{Contribution}
We establish a new equivalence principle between the
problem of demand inversion and the problem of stable matchings in two-sided
models with Imperfectly Transferable Utility (ITU). More precisely, we show
that a discrete choice model can always be interpreted as a two-sided matching market
where consumers and alternatives are viewed as firms and workers; and that
the \emph{demand inversion problem}, that is the recovery of utility
vectors $\left( \delta _{j}\right)_{j\in \mathcal{J}}$, can be reformulated
as the \emph{equilibrium problem} of determining competitive wages in the
corresponding  matching market. In other words, the identified set of solution
vectors $\delta $ coincides with the set of equilibrium wages in the
matching market. This equivalence implies two important contributions:

\begin{enumerate}
\item \textbf{Characterization of the identified set of $\delta$.} The equivalence to the matching equilibrium implies that the identified set of vectors $\delta _{j}$ is a \emph{lattice}, from which one can construct a very simple data-driven test for point-identification.\footnote{A lattice is a partially ordered set that contains the meet and the join of each pair of its element. For the purposes of this paper, lattices are subsets of vectors in Euclidean space and, for any given pair of vectors, the meet (resp. join) is just the vector containing the componentwise infimum (resp. supremum). For additional discussion of lattices in matching theory, consult \cite{roth1992two}. Relatedly, \citeasnoun{jia2008happens} exploits the lattice structure of equilibria in supermodular games to estimate a large multi-market entry game between discount retailers.} As such, if the greatest element of the lattice
coincides with its smallest element, then the utility index $\delta$ is point-identified.\footnote{\citeasnoun{khan2016identification} call this an \textquotedblleft adaptive\textquotedblright\ property.}  Thus, our approach bypasses the need of verifying {\em a priori} whether the parameters of a given model are point of partial-identified\textemdash a non-trivial exercise in many cases.  

\item \textbf{Computation of the identified set of $\delta$.}
Our matching approach has two key features. First, the matching equivalence allows the utilization of several high-performance matching algorithms, for which the convergence properties are well-studied. The use of these matching algorithms for estimating random utility models is new; moreover, they can readily handle situations in which multiple values of $\delta$ rationalize the observed market shares.   Second, these matching-based algorithms do not require the computation of the demand (market-share) mapping. This is important in specific models, such as the {\em pure characteristics model} (\citeasnoun{berry2007pure}), which are notorious for their non-smooth market-share mappings.  Indeed, using our approach, the demand inversion problem for the pure characteristics model becomes a well-behaved {\em convex program} (see Section 5 below).
\end{enumerate}

\subsection{Related literature}\label{sec:literature}
{ 
In this paragraph we discuss how our paper relates to the existing literature on (1) demand inversion and (2) two-sided matching.

\textbf{Demand inversion literature.} Demand inversion is a crucial intermediate step for estimating aggregate discrete-choice models of product-differentiated markets; see, e.g., \citeasnoun{berry1994estimating} and \citeasnoun{berry1995automobile} (BLP). It also plays an important role in two-step estimation procedures for dynamic discrete-choice models\footnote{ Specifically, even though we don't pursue the connection here, the demand inversion problem considered in this paper also applies straightforwardly to the problem in dynamic discrete-choice estimation of ``inverting" the choice-specific value functions from the conditional choice probabilities.} (including \citeasnoun%
{hotz1993conditional},  \citeasnoun{aguirregabiria2002swapping}, \citeasnoun{bajari2007estimating}, \citeasnoun{arcidiacono2011conditional}, \citeasnoun{kristensen2014ccp}).

Theoretically, there is a large literature that tackles the problem of the invertibility of the demand map utilizing either \emph{differentiability} or \emph{monotonicity} of the demand map.  The first approach, based on global univalence
theorems, requires the \emph{differentiability} of the mapping $\sigma \left(
\delta \right) $.
\footnote{The Hadamard-Palais univalence theorem (\citeasnoun{palais1959natural})
asserts that (i) if $\sigma :\mathbb{R}^{J}\rightarrow \mathbb{R}^{J}$
is $C^{1}$, (ii) if its Jacobian is invertible at all points, and (iii) if $\left\Vert
\sigma \left( \delta \right) \right\Vert \rightarrow \infty $ as $\left\Vert
\delta \right\Vert \rightarrow \infty $, then $\sigma $ is globally
invertible; further results are collected in~\citeasnoun{Parthasarathy} and~\citeasnoun{RadulescuRadulescu}.} 
Such results were used by \citeasnoun{chiappori} and~\citeasnoun{kristensen2014ccp} to study (resp.) multinomial choice and  nonadditive random utility models.
The results in~\citeasnoun{GaleNikaido} focus on
uniqueness, leaving existence aside: assuming that the Jacobian of $\sigma $
is a P-matrix and that the domain is a rectangle, Gale-Nikaido's result
guarantees the injectivity of $\sigma $\textemdash namely, that $\sigma ^{-1}\left(
\left\{ s\right\} \right) $ should have at most one point\textemdash but can be empty.

A second approach to demand inversion relies on \emph{gross substitutes}\textemdash a form of \emph{monotonicity}\textemdash that $\sigma _{j}$ is decreasing, or at least nonincreasing, with respect to
$\delta _{j^{\prime }}$. This category of papers includes the literature on ARUMs, where this property holds automatically (eg. \citeasnoun{mcfadden1978modelling}). \citeasnoun%
{hotz1993conditional} study the problem of the nonemptiness
of $\sigma ^{-1}\left( s\right)$, within general ARUMs ; \citeasnoun{berry1994estimating} provided a
complete argument (which extends to nonadditive random utility
models), and also shows uniqueness of the vector of systematic
utilities under continuity conditions. \citeasnoun{magnac2002identifying} investigate identification of structural parameters (period utility flows)
in dynamic discrete choice models.  \citeasnoun{norets2013surjectivity} focus on the surjectivity of ARUMs, under the assumption of
absolute continuity of the distribution of the additive utility shocks.
More broadly, \citeasnoun{berry2013connected} show injectivity in general demand systems with gross substitutes (thus going beyond random utility models) under a connected strong substitute assumption. 

To date, the existing literature has not provided a general characterization of
the identified utility set $\sigma^{-1}(s)$, defined as the set of utility vectors $\left( \delta _{j}\right) $ which rationalize
a vector of market shares $\left( s_{j}\right)$. 
This paper is the first to consider situations when the identified utility set
$\sigma^{-1}\left( \left\{ s\right\} \right) $ is non-singleton.   In addition, most of the papers cited above provide little guidance on computing the identified set.\footnote{As
\citeasnoun[p. 10]{berry2015identification} underline, \textquotedblleft
(...) the invertibility result of~\citeasnoun{berry2013connected} is not a
characterization (or computational algorithm) for the
inverse\textquotedblright.}  Besides a handful of models,\footnote{These include the logit, nested logit, and random-coefficient logit models (\citeasnoun{berry1994estimating}, \citeasnoun{berry1995automobile}, \citeasnoun{dube2012improving}).} there are no well-established procedures for demand-inversion in general (non-additive) random utility models with arbitrary error distributions.   Our paper aims to fill this gap. In doing so, it builds on a set of recent papers which have reformulated the problem of demand inversion in ARUMs as an optimal transport problem, using the tools of convex duality.  This approach was pioneered by~\citeasnoun{galichon2017cupid}, and was extended to ARUMs with possibly noncontinuous distributions of unobserved heterogeneity by \citeasnoun{chiong2015duality},
and to continuous choice problems by~\citeasnoun{CGHP}.  However, these papers do not cover nonadditive random utility models, and do not characterize the structure of the identified set.

\textbf{Two-sided matching literature.} We provide a brief and incomplete review of the two-sided matching literature, as one key result in this paper is to show the equivalence between demand inversion and a  two-sided matching model.
The matching literature is split between models with non-transferable utility (NTU), and
transferable utility (TU), with intermediate cases called imperfectly transferable utility (ITU).\footnote{A key reference for the matching literature (TU, NTU and ITU alike) is~\citeasnoun{roth1992two},
while~\citeasnoun{galichon2015optimal} focuses on the TU case, or equivalently, optimal transport methods.}
A connection was made earlier between TU matching models and ARUMs (see~\citeasnoun{galichon2017cupid} and~\citeasnoun{chiong2015duality}), and in the present paper, we are making a novel connection between matching models with ITU and NARUMs.

A large class of algorithms for computing matching models consists of ``deferred acceptance algorithms'' which revolve around Tarski's fixed point theorem, and interpreting stable matchings as fixed points of monotone mappings.  They apply to NTU and ITU models. 
These ideas appeared in~\citeasnoun{Adachi2000},
followed by~\citeasnoun{Fleiner}, \citeasnoun{EcheniqueOvideo} and~\citeasnoun{hatfield2005matching}. The
seminal deferred acceptance algorithms (\citeasnoun{gale1962college}; \citeasnoun{crawford1981job}; \citeasnoun{kelso1982job}) can be interpreted in this way.

Other methods are ``descent methods,'' which rely on a reformulation of the problem as a convex optimization problem, and apply in the TU case. Some of the methods involve coordinate descent (as the auction algorithm of~\citeasnoun{bertsekas1989auction}), while some involve gradient descent (as the semi-discrete approach of~\citeasnoun{aurenhammer1987power}). The linear programming solution to the optimal assignment problem described in~\citeasnoun{shapley1971assignment} also belongs in this category.
The study of rates of convergence of these algorithms have been the subject of intense study; see the book~\citeasnoun{peyrecuturi} and~\citeasnoun{Book_assignment} for results and further references. }

\subsection{Organization}

Section~\ref{the framework} introduces the general random utility framework
which is the focus of this paper and provides examples. 
Section \ref%
{sec:DemandInversionNoTies} presents our main equivalence result between NARUMs and two-sided matching problems, and discusses the lattice structure of the identified
utility set. Based on the equivalence result, in Section %
\ref{sec:algorithms} we introduce several matching-based algorithms which can solve a wide
variety of random utility models.  Section \ref{sec:numerical experiments}
contains two simulation investigations of the algorithms, including the pure characteristics model.  Section \ref{sec:empirical} utilizes our matching approach to estimate a spatial voting model using electoral data from the 1999 European Parliament elections. 
Section~\ref{sec: conclusion} concludes.  
Proofs and several additional theoretical results are collected in the appendix.

\section{The framework}

\label{the framework}

\label{NARUM : presentation of the framework}

\subsection{Basic assumptions}

Let $\mathcal{J}_{0}=\mathcal{J}\cup \left\{ 0\right\} $ be a finite set of
alternatives, where $j=0$ denotes a special alternative which serves as a benchmark (see section \ref{sec: normalization} below).
The agent's program is thus%
\begin{equation}
u_{\varepsilon _{i}}=\max_{j\in \mathcal{J}_{0}}\left\{ \mathcal{U}%
_{\varepsilon _{i}j}\left( \delta _{j}\right) \right\} ,
\label{def:uepsilon}
\end{equation}%
where $u_{\varepsilon _{i}}$ is the indirect utility of an agent with shock $%
\varepsilon _{i}$. The utility agent $i$ derives from alternative $j$
depends on the systematic utility vector $\delta _{j}$ associated with this
alternative, and on the realization $\varepsilon _{i}$\ of this agent's
utility shock. We will work under two assumptions:

\begin{assumption}[Regularity of $\mathcal{U}$]
\label{ass:Increasing}Assume $\left( \Omega ,P\right) $ is a Borel
probability space and for every $\varepsilon \in \Omega $, and for every $%
j\in \mathcal{J}_{0}$:

(a) the map\ $\varepsilon \mapsto \left( \mathcal{U}_{\varepsilon j}\left(
\delta _{j}\right) \right) _{j\in \mathcal{J}_{0}}$ is measurable, and

(b) the map $\delta _{j}\mapsto \mathcal{U}_{\varepsilon j}\left( \delta
_{j}\right) $ is increasing from $\mathbb{R}$ to $\mathbb{R}$ and continuous.
\end{assumption}

\begin{assumption}[No indifference]
\label{ass:NoIndiff}

For every distinct pair of indices $j$ and $j^{\prime }$
in $\mathcal{J}_{0}$, and for every pair of scalars $\delta $ and $\delta
^{\prime }$,
\begin{equation*}
P\left( \varepsilon \in \Omega :\mathcal{U}_{\varepsilon j}\left( \delta
\right) =\mathcal{U}_{\varepsilon j^{\prime }}\left( \delta ^{\prime
}\right) \right) =0.
\end{equation*}
\end{assumption}

These two assumptions are standard in the literature, and are automatically satisfied in ARUMs. Assumption \ref{ass:Increasing} (a) is a standard measurability condition, and (b) is often invoked in the literature on non-separable models (see, e.g., \citeasnoun{Matzkin2007}).\footnote{Assumption 1 also circumscribes the set of models we consider here.  Specifically, it rules out purely "horizontal"  (Hotelling) choice models, where  the utility of consumer $\varepsilon$ for store $j$ is $\mathcal{U}_{\varepsilon j}(\delta_j) = -|\delta_{j}-\varepsilon|$ , interpreted as the (minus the) absolute distance between the consumer's home location $\varepsilon$ and the store's location  $\delta_j$ where $\delta_j, \varepsilon \in [0,1]$.  
This specification violates the monotonicity condition (Assumption 1(b)).  Interestingly, it turns out the {\em quadratic} version of this model can be reparametrized in a way which satisfies our modelling assumptions; see Section 6 below.}

Assumption \ref{ass:NoIndiff} rules out indifference (precisely, an event of measure zero) between two alternatives, and is maintained in practically all the applied discrete choice literature.  
{In earlier versions of the paper (available from the authors upon request), we showed that the results in this paper hold even without
Assumption \ref{ass:NoIndiff}, albeit at the greater notational expense of introducing
set-valued functions.} 
In the present paper, for simplicity, we maintain Assumption~\ref{ass:NoIndiff} as it suffices for our purposes. 

Under Assumption \ref{ass:NoIndiff}, the demand of
alternative $j$ (defined in (\ref{ProbaDetermination}) above) corresponds to the fraction of consumers who prefer
weakly \emph{or} strictly alternative $j$ to any other one:
\begin{equation}
\sigma _{j}\left( \delta \right) :=P\left( \varepsilon \in \Omega :\mathcal{U%
}_{\varepsilon j}\left( \delta _{j}\right) \geq \max_{j^{\prime }\in
\mathcal{J}_{0}}\mathcal{U}_{\varepsilon j^{\prime }}\left( \delta
_{j^{\prime }}\right) \right) =P\left( \varepsilon \in \Omega :\mathcal{U}%
_{\varepsilon j}\left( \delta _{j}\right) >\max_{j^{\prime }\in \mathcal{J}%
_{0}\backslash \left\{ j\right\} }\mathcal{U}_{\varepsilon j^{\prime
}}\left( \delta _{j^{\prime }}\right) \right) .  \label{defMktShare}
\end{equation}%
Importantly, the uniqueness of the vector of market shares associated with a given utility vector $\delta $ does not imply that the demand inversion problem has a unique solution. There may be multiple vectors $\delta $ such that $%
\sigma (\delta )=s$.
Under Assumptions~\ref{ass:Increasing} and~\ref{ass:NoIndiff}, the vector of
market shares $s=\sigma \left( \delta \right) $ is a probability vector on $%
\mathcal{J}_{0}$, which prompts us to introduce $\mathcal{S}_{0}$, the set
of such probability vectors as
\begin{equation*}
\mathcal{S}_{0}:=\left\{ s\in \mathbb{R}_{+}^{\mathcal{J}_{0}}:\sum_{j\in
\mathcal{J}_{0}}s_{j}=1\right\} .
\end{equation*}

We formalize the definition of the demand map.

\begin{definition}[Demand map]
\label{def:MktShareMap}Under Assumption~\ref{ass:Increasing} and \ref{ass:NoIndiff}, the \emph{demand
map} is the map~$\sigma :\mathbb{R}^{\mathcal{J}_{0}}\rightarrow \mathcal{S}%
_{0}$ defined by expression~(\ref{defMktShare}).
\end{definition}

\subsection{Normalization}

\label{sec: normalization}

Any discrete choice model requires some normalization, because the choice
probabilities result from the comparison of the relative utility payoffs
from each alternative. Throughout the paper we normalize the
systematic utility associated to the default alternative to zero:
\begin{equation}
\delta _{0}=0,  \label{normalization}
\end{equation}%
and we use
\begin{equation}
\tilde{\sigma}:\mathbb{R}^{\mathcal{J}}\rightarrow \mathbb{R}^{\mathcal{J}}
\label{defMktShareNormalized}
\end{equation}%
to denote the map induced by this normalization.

In the special ARUM case where $\mathcal{U}%
_{\varepsilon _{i}j}\left( \delta _{j}\right) =\delta _{j}+\varepsilon _{ij}$%
, imposing normalization~(\ref{normalization}) is innocuous because the vector of systematic utilities $\left( \delta _{j}\right) $ yields
the same choice problem as the vector $\left( \delta _{j}+c\right) $
where $c$ is a constant; hence, for ARUMs, any normalization will
yield the same identified utility vectors $\delta $ up to an additive
constant. However, this is no longer true in nonadditive models, for which the
normalization~(\ref{normalization}) entails some loss of generality.\footnote{While it is possible to reparameterize any NARUM into an ARUM $\mathcal{U}_{\varepsilon j}(\delta)=\gamma_j+\eta_j$ by defining $\gamma_j\equiv \int \mathcal{U}_{\varepsilon j}(\delta) dP_{\varepsilon}$ and $\eta_j= \mathcal{U}_{\varepsilon j}(\delta)-\gamma_j$, it is not clear whether this reparameterization simplifies the demand inversion problem of recovering $\delta$, which is the main point of our paper (in some cases, the reparameterization obfuscates the demand inversion problem).} We explore this below in Section 5.

\subsection{Examples}

Next, we consider several examples of random utility models falling within
our framework.  Since we assume that market shares are generated by the mapping in Eq. (\ref{defMktShare}), we implicitly assume that the random element $\varepsilon$ is independently and identically distributed across all consumers in the market.   However, we make no restrictions on $\varepsilon$ across alternatives: as the examples below show, $\varepsilon$ can be individual-specific, choice-specific, or some combination of the two.  

\begin{example}[ARUM]
\label{ex:ARUM}In the additive random utility model (ARUM) one sets, $\Omega
=\mathbb{R}^{\mathcal{J}_{0}}$, so that $P$ is a probability distribution on
$\mathbb{R}^{\mathcal{J}_{0}}$, and%
\begin{equation*}
\mathcal{U}_{\varepsilon _{i}j}\left( \delta _{j}\right) =\delta
_{j}+\varepsilon _{ij}.
\end{equation*}

There are several well-known instances of ARUMs:

\underline{Logit model:} if $P$ is the distribution of a vector of size $%
\left\vert \mathcal{J}_{0}\right\vert $ of i.i.d. type 1-Extreme value
random variables, then the demand map is
given by $\sigma _{j}\left( \delta \right) =\exp \left( \delta _{j}\right)
/\left( \sum_{j^{\prime }\in \mathcal{J}_{0}} \exp(\delta _{j^{\prime }})\right) $.\footnote{For the logit model, the use of aggregate shares in forming moment conditions has been long recognized in the literature, going back to \citeasnoun{berkson1955maximum}; \citeasnoun{berry1994estimating} and \citeasnoun{berry1995automobile} showed that this connection applies more generally to static discrete-choice models beyond the simple logit model.}
The demand map is analytically invertible, and yields the familiar 
\textquotedblleft log-odds ratio\textquotedblright\ formula:

\begin{equation}
\delta_{j}=\log \left( s_{j}/s_{0}\right) .
\label{logOddsRatio}
\end{equation}


\underline{Pure characteristics model:} In this model, consumers value product $j$ only through its
measurable characteristics $x_{j}\in \mathbb{R}^{d}$, a vector of dimension $%
d$ associated to each alternative $j$, and the utility shock vector $%
\varepsilon _{i}$ is such that
\begin{equation}
\varepsilon _{ij}=\nu _{i}^{\intercal }x_{j}=\sum_{k=1}^{d}\nu
_{i}^{k}x_{j}^{k}\label{utility_purechar}
\end{equation}%
where $\nu _{i}$ is consumer $i$'s vector of taste-shifters, drawn from a
distribution $P_{\nu }$ on $\mathbb{R}^{d}$. In this case, there is no
closed-form expression for the demand map\footnote{See \citeasnoun{song2007measuring} and \citeasnoun{nosko2010competition} for two empirical applications of the pure characteristics demand model.  \citeasnoun{pang2015constructive} provide computational algorithms for estimating this model.}. We shall revisit this model in
the simulations and empirical application in sections \ref{sec:numerical experiments} and \ref{sec:empirical}.

The case where $d=1$ and there is only one characteristic, the price $p_j$, is the {\em vertical differentiation} or {\em quality ladder} model\footnote{%
See, among others, \citeasnoun{prescott1977sequential}, %
\citeasnoun{bresnahan1981departures}, \citeasnoun{esteban2007durable}.}, which has the utility specification 
\begin{equation*}
\mathcal{U}_{\varepsilon j}(\delta_j) = \delta _{j}-\nu_i p_{j},\quad \forall j.
\end{equation*}%
Here $\delta _{j}$ is interpreted as the quality of brand $j$, while the
nonlinear random utility shock $\nu _{i}$ measures household $i$'s
willingness-to-pay for quality. Below, in Section \ref{sec:numerical
experiments}, we consider a numerical example based on this framework
which is non-invertible.

\underline{Random coefficient logit model:} In the random coefficient logit
model popularized by BLP and~\citeasnoun{mcfadden2000mixed}, the random-utility shock is given by:%
\begin{equation*}
\varepsilon _{ij}=\nu _{i}^{\intercal }x_{j}+ \zeta _{ij}.
\end{equation*}%
This is the sum of two independent terms: one logit term $\zeta _{ij}$ and one pure characteristics
term $\nu _{i}^{\intercal }x_{j}$.  

\end{example}

\begin{example}[Risk aversion]
\label{Example 2: Discrete choice under uncertainty}Consider a market where
consumers are not fully aware of the attributes of a product at the time of
purchase. This may characterize consumers' choices in online markets, where
they have no opportunity to physically examine the goods under
consideration. Let $\varepsilon _{i}$ denote the relative risk aversion
parameter (under CRRA utility), and that the price of good $j$ is $p_{j}$.
Choosing option $j$ yields a consumer surplus of $\delta _{j}-p_{j}+\eta
_{j}$ where $\log \eta _{j}\sim N(0,1)$ is a quality shock unobservable at
the time of the purchase, and $\delta _{j}$ is the willingness to pay (in
dollar terms) associated to alternative $j$. At the time of the purchase,
the consumer's expected utility is%
\begin{equation*}
\mathcal{U}_{\varepsilon _{i}j}\left( \delta _{j}\right) =\mathbb{E}_{\eta
_{j}}\left[ \frac{\left( \delta _{j}-p_{j}+\eta _{j}\right) ^{1-\varepsilon
_{i}}}{1-\varepsilon _{i}}\right] ,
\end{equation*}%
where the expectation is taken over $\eta _{j}$ holding $\varepsilon _{i}$
constant. These kind of models are typically non-additive in $\varepsilon$.\footnote{See \citeasnoun{cohen2007estimating} and \citeasnoun{apesteguia2014discrete} for examples.}
\end{example}

\section{Equivalence of Discrete-Choice and Two-Sided Matching\label%
{sec:DemandInversionNoTies}}

In this section, we show a central result of this paper; namely, an equivalence between discrete-choice models and two-sided matching problems.
This equivalence is noteworthy as discrete choice problems are traditionally considered to be ``one-sided'' problems.
However, we will demonstrate that they are equivalent to a two-sided
\textquotedblleft marriage problem\textquotedblright\ between
consumers and yogurts, where both sides of the market must assent to be
matched. 

{Specifically, our main theorem demonstrates an equivalence between the discrete-choice framework described in the previous section and a two-sided matching market with imperfectly transferable utility 
(see~\citeasnoun{galichon2014empirical}). 
}
A consequence of this equivalence is that the demand inversion problem for estimating discrete-choice models can be equivalently formulated as solving for equilibrium utility payoffs from the corresponding two-sided matching problem, a well-understood exercise for which a number of algorithms are available.  


We begin by formally defining the object of interest for demand inversion, which is to recover the identified utility set:
\begin{definition}[Identified utility set]
\label{def:identifiedUtSet}Given a demand map $\tilde{\sigma}$ defined as
in~(\ref{defMktShareNormalized}) where Assumptions~\ref{ass:Increasing} and~%
\ref{ass:NoIndiff} are met, and given a vector of market shares $s$ that
satisfies $s_{j}>0$ and$~\sum_{j\in \mathcal{J}_{0}}s_{j}=1$, the \emph{%
identified utility set} associated with $s$ is defined by%
\begin{equation}
\tilde{\sigma}^{-1}\left( s\right) =\left\{ \delta \in \mathbb{R}^{\mathcal{J%
}}:\tilde{\sigma}\left( \delta \right) =s\right\} .  \label{identifiedSet}
\end{equation}
\end{definition}

\noindent Requiring  non-zero market shares is a standard assumption in demand inversion of discrete-choice models; see, e.g., Lemma 1 of \citeasnoun{berry2014identification}.\footnote{ In the ARUM case, a sufficient condition for this is that $\varepsilon$ has full-support (a nowhere vanishing density) on $\mathbb{R}^{\mathcal{J}_{0}}$ (see, e.g., \citeasnoun{GalichonHsieh2019}), which is satisfied in models such as logit, probit, and mixed-logit.  However, it may fail to hold in, e.g.,  pure characteristic models (see, e.g., \citeasnoun{song2007measuring}), without sufficient variation in the choice-specific characteristics. 
Moreover, we assume throughout that the demand model is correctly specified, so that the identified utility set in (\ref{identifiedSet}) is non-empty. 
}

\subsection{The Equivalence Theorem\label{par:noIndiff}}
Next we introduce a matching game between consumers and yogurts, which is essentially that of~\citeasnoun{demange1985strategy}; our
presentation of this model is inspired by the presentation in chapter 9 of~%
\citeasnoun{roth1992two}\footnote{%
Demange and Gale's model is discrete and extends the model of %
\citeasnoun{shapley1971assignment} beyond the transferable utility setting.
See also~\citeasnoun{crawford1981job}, \citeasnoun{kelso1982job}, %
\citeasnoun{hatfield2005matching}. We formulate a slight variant here in
that we (1) we allow for multiple agents per type and (2) do not allow for
unmatched agents. However, this leaves analysis essentially unchanged.}.
In this matching model, one side of the market consists of a continuum of consumers, distinguished by type $\varepsilon$, while the other side is an equi-massed continuum of jars of yogurt, distinguished by alternative (brand) $j\in\mathcal{J}_0$.
We start by introducing some terminology.  Let $\mathcal{M}\left( P,s\right) $ be the set of probability distributions on $%
\Omega \times \mathcal{J}_{0}$ with marginal distributions $P$ and $s$;
namely, $\pi \in \mathcal{M}\left( P,s\right) $ if and only if $\pi \left(
B\times \mathcal{J}_{0}\right) =P\left( B\right) $ for all $B$ (Borel-measurable
subsets of $\Omega$), and $\pi \left( \Omega \times \left\{ j\right\} \right)
=s_{j}$ for all $j\in \mathcal{J}_{0}$.

{Let $u_{\varepsilon}$ (resp. $v_j$) denote utility payoffs for each consumer $\varepsilon$ (resp. yogurt $j$) from the matching game.  These payoffs are endogenously determined in equilibrium, as will be clear below.}
Let $f_{\varepsilon j}\left(
u\right) $ be the transfer (positive or negative) needed by a consumer $%
\varepsilon $ in order to reach utility level $u\in \mathbb{R}$ when matched
with a yogurt $j$. Symmetrically, let $g_{\varepsilon j}\left( v\right) $ be
the transfer needed by a yogurt $j$ in order to reach utility level $v\in
\mathbb{R}$ when matched with a consumer $\varepsilon $.
(The connection between $f$ and $g$ and the primitives of the discrete choice model will be clarified below, in Eq.~\eqref{conversion}.) The functions $%
f_{\varepsilon j}\left( .\right) $ and $g_{\varepsilon j}\left( .\right) $
are assumed increasing for every $\varepsilon $ and $j$.  

{
We describe this game in the context of a dance party where consumers (indexed by $\varepsilon$)  seek out jars of yogurt $j$ to dance with.
%
Each consumer $\varepsilon$ charges a price of $u_{\varepsilon}$ utils for dancing; similarly, yogurt $j$ charges a price of  $v_j$ for dancing.  Thus, a dance between consumer $\varepsilon$ and yogurt $j$ involves a payment of $g_{\varepsilon j}(v_j)$ from consumer $\varepsilon$ to yogurt $j$, and a payment of $f_{\varepsilon j}(u_\varepsilon)$ from the yogurt to the consumer.  These transfers can be negative: for instance, if yogurt $j$ already provides consumer $\varepsilon$ with utility exceeding $u_{\varepsilon}$ by dancing with her, then $f_{\varepsilon j}(u_\varepsilon)<0$ and involves a payment from consumer $\varepsilon$ to yogurt $j$. 
}

\begin{definition}[Equilibrium outcome]
\label{def:equilibrium}An equilibrium outcome in the matching problem is an
element $\left( \pi ,u,v\right) $, where $\pi $ is a joint probability measure on $%
\Omega \times \mathcal{J}_{0}$, $u$ and $v$ are Borel-measurable functions on $%
\left( \Omega ,P\right) $ and $\left( \mathcal{J}_{0},s\right) $
respectively, such that:

(i) $\pi $ has marginal distributions $P$ and $s$: $\pi \in \mathcal{M}%
\left( P,s\right) $.

(ii) there is no blocking pair: $f_{\varepsilon j}\left( u_{\varepsilon
}\right) +g_{\varepsilon j}\left( v_{j}\right) \geq 0$ for all $\varepsilon
\in \Omega $ and $j\in \mathcal{J}_{0}$.

(iii) pairwise feasibility holds: if $\left( \varepsilon ,j\right) \in
Supp\left( \pi \right) $, then $f_{\varepsilon j}\left( u_{\varepsilon
}\right) +g_{\varepsilon j}\left( v_{j}\right) =0$.
\end{definition}

  Condition (i) implies that if a random vector $\left( \varepsilon ,j\right) $ has
distribution $\pi \in \mathcal{M}\left( P,s\right) $, then $\varepsilon \sim
P$ and $j\sim s$. Hence, $\pi $ is interpreted as the probability
distribution that a consumer with utility shock $\varepsilon $ is matched
with a yogurt of type $j$; in other words, $\pi \left( j|\varepsilon \right) $
denotes the conditional probability that an individual with utility shock $%
\varepsilon $ chooses yogurt $j$, which is degenerate (equaling 0 or 1) under Assumption 2. To understand condition (ii), 
consider that if there exists a consumer
$\varepsilon $ and a yogurt of type $j$ for which
$f_{\varepsilon j}\left( u_{\varepsilon }\right) +g_{\varepsilon
j}\left( v_{j}\right) <0$, then there exists $u^{\prime }>u_{\varepsilon }$
and $v^{\prime }>v_{j}$ such that $f_{\varepsilon j}\left( u^{\prime
}\right) +g_{\varepsilon j}\left( v^{\prime }\right) =0$. In other words, $(u', v')$ are feasible for $\left( \varepsilon
,j\right) $ and strictly improve upon the equilibrium payoffs $%
u_{\varepsilon }$ and $v_{j}$, which is ruled out in equilibrium. { The feasibility condition (iii) rules out matchings involving net positive transfers ($f_{\varepsilon j}(u_{\varepsilon})+g_{\varepsilon j}(v_j)>0$) which, intuitively, are those where one (or both) of the agents have equilibrium payoffs which are inachievably high compared to the utility they supply each other from matching.
}

The next theorem establishes that the demand inversion problem is
equivalent to a matching problem. The proofs for this and all subsequent
claims are in the appendix.

\begin{theorem}[Equivalence theorem]
\label{thm:equivalenceNoTies}Under Assumptions~\ref{ass:Increasing} and~\ref%
{ass:NoIndiff}, consider a vector of market shares $s$ that satisfies $%
s_{j}>0$ and$~\sum_{j\in \mathcal{J}_{0}}s_{j}=1$. Consider a vector $\delta
\in \mathbb{R}^{\mathcal{J}}$. Then, the two following statements are
equivalent:

(i) $\delta $ belongs to the identified utility set $\tilde{\sigma}%
^{-1}\left( s\right) =\left\{ \delta \in \mathbb{R}^{\mathcal{J}}:\tilde{%
\sigma}\left( \delta \right) =s\right\} $ associated with the market shares $%
s$ in the sense of Definition~\ref{def:identifiedUtSet} in the discrete
choice problem with $\varepsilon \sim P$;

(ii) there exists $\pi \in \mathcal{M}\left( P,s\right) $ and $%
u_{\varepsilon}=\max_{j\in \mathcal{J}_{0}} \mathcal{U}_{\varepsilon
j}\left( \delta _{j}\right)$ such that $\left( \pi ,u,-\delta \right) $ is
an equilibrium outcome in the sense of Definition~\ref{def:equilibrium} in
the matching problem with transfer functions
\begin{equation}
f_{\varepsilon j}\left( u\right) =u\text{ and }g_{\varepsilon j}\left(
-\delta \right) =-\mathcal{U}_{\varepsilon j}\left( \delta \right) .
\label{conversion}
\end{equation}
\end{theorem}


We offer some intuition of the equivalence here. 
A matching equilibrium requires that, given the transfer functions, the equilibrium transfers be set such that {\em both} sides of the market (consumers and yogurts) are happy with their matched partners.   On the consumer side,
consumer $%
\varepsilon $ seeks the yogurt $j$ which offers her the highest payoff $\mathcal{U}_{\epsilon}$.   On the yogurt side, each yogurt $j$ also seeks to maximize its payoff, which is equivalent to {\em minimizing} the mean utility $\delta_j$ that consumers receive from them; for that reason, we have the yogurt's payoff $v_j=-\delta_j$.  To understand the intuition for the transfer functions (\ref{conversion}), note that if we plug these into  Definition 3, the no-blocking pair condition becomes
\begin{equation}\label{eq:nbp1}
\forall\ \varepsilon,\ j: \quad u_{\varepsilon }=\max\limits_{j^{\prime }\in
\mathcal{J}_{0}}\mathcal{U}_{\varepsilon j^{\prime }}\left( \delta
_{j^{\prime }}\right) \geq \mathcal{U}_{\varepsilon j}\left( \delta
_{j}\right);
\end{equation}
and the feasibility pair condition becomes
\begin{equation}\label{eq:nbp2}
\text{if $%
\left( \varepsilon ,j\right) \in Supp(\pi ):$}\quad  u_{\varepsilon
}=\max\limits_{j^{\prime }\in \mathcal{J}_{0}}\mathcal{U}_{\varepsilon
j^{\prime }}\left( \delta _{j^{\prime }}\right) =\mathcal{U}_{\varepsilon
  j}\left( \delta _{j}\right)
\end{equation}
which correspond to the conditions characterizing optimal consumer choices in the discrete-choice problem.

Our equivalence result states that the identified set of utilities in the discrete-choice demand problems corresponds to the equilibrium set of some matching problem. This matching equivalence result also raises the possibility of multiplicity of the identified set of utilities.  Assumption \ref{ass:NoIndiff} implies a unique demand map (Eq. (\ref{ProbaDetermination})), and hence a unique allocation in the matching problem.  However, just as in \citeasnoun{shapley1971assignment}, there may be multiple payoffs (corresponding to the $\delta$'s here) which support the equilibrium allocation.  We will return to this below.

In the special case when the random utility model is additive (ARUM) as in example~\ref{ex:ARUM}, one has $f_{\varepsilon j}\left( u\right) =u$ and $g_{\varepsilon j}\left(
-\delta \right) =-\mathcal{U}_{\varepsilon j}\left( \delta \right)
=-\varepsilon _{j}-\delta _{j}$, so that the stability conditions become $%
u_{\varepsilon }+v_{j}\geq \varepsilon _{j}$ with equality for $\left(
\varepsilon ,j\right) \in Supp(\pi )$. As noted initially in~%
\citeasnoun{galichon2017cupid}, this problem is now equivalent to a matching problem
with transferable utility, where the joint surplus of a match between a
consumer $\varepsilon $ and a yogurt $j$ is $\varepsilon _{j}$.  We will discuss these in more detail in section \ref{sec:algo_arum} below.

\subsection{Lattice structure of the identified utility set\label{par:latticeStructure}}

Next, we show that the set-valued function $s\rightarrow \tilde{\sigma}^{-1}\left( s\right) $
is isotone\footnote{Isotone has the meaning of (monotone) increasing, and is a standard term used in the literature on lattices (see \citeasnoun{topkissupermodularity} for a reference on lattices and isotonicity in economics).} (in a sense to be made precise) and that $\tilde{\sigma}^{-1}\left(
s\right) $ has a lattice structure. 

The lattice structure is very useful for deriving algorithms for computing the identified utility set $\tilde{\sigma}^{-1}\left(
s\right) $ when demand is non-invertible.   The literature on the estimation of discrete choice models has favored an
approach based on imposing conditions guaranteeing invertibility of demand,
or equivalently situations in which $\tilde{\sigma}^{-1}\left( s\right) $ is
restricted to a single point. In particular, \citeasnoun{berry2013connected} (hereafter BGH)
provide conditions under which $\tilde{\sigma}^{-1}\left( s\right) $ should
contain at most one point, from which it also follows that the map $%
s\rightarrow \tilde{\sigma}^{-1}\left( s\right) $ is isotone on its domain.
In contrast, our approach here imposes minimal assumptions, and one must
consider {\em non-invertible} models, in which the demand map (\ref{ProbaDetermination}) is
not one-to-one and $\tilde{\sigma}^{-1}\left( s\right) $ is a set. In this
case, we need to generalize the notion of isotonicity which applies to the identified set $\tilde{\sigma}^{-1}(.)$. The next theorem states that the correct generalization is the notion of
isotonicity with respect to {\em Veinott's strong set order}.\footnote{See e.g. %
  \citeasnoun{veinottlectures}.  Veinott's strong set order provides an ordering over sets.   Let $\mathcal{X}$ and $\mathcal{X}'$ be two subsets in $\mathbb{R}^d$; we say that $\mathcal{X} < \mathcal{X}'$ in the Veinott strong set order iff $\forall x\in \mathcal{X}, x'\in\mathcal{X}'$, the ``join'' (or componentwise minimum) $x\wedge x' \in \mathcal{X}$ and the ``meet'' (or componentwise maximum) $x\vee x' \in \mathcal{X}'$.}
For the following, recall
the lattice \textquotedblleft join\textquotedblright\ and \textquotedblleft
meet\textquotedblright\ operators ($\wedge $ and $\vee $) are defined by $%
\left( \delta \wedge \delta ^{\prime }\right) _{j}:=\min \left\{ \delta
_{j},\delta _{j}^{\prime }\right\} $ (componentwise minimum) and $\left( \delta \vee \delta ^{\prime
}\right) _{j}:=\max \left\{ \delta _{j},\delta _{j}^{\prime }\right\} $ (componentwise maximum).

\begin{theorem}
\label{thm:InverseIsotoneNoTies}The set-valued function $s\rightarrow \tilde{%
\sigma}^{-1}\left( s\right) $ is isotone in Veinott's strong set order, i.e.
if $\delta \in \tilde{\sigma}^{-1}\left( s\right) $ and $\delta ^{\prime
}\in \tilde{\sigma}^{-1}\left( s^{\prime }\right) $ with $s\leq s^{\prime }$%
, then $\delta \wedge \delta ^{\prime }\in \tilde{\sigma}^{-1}\left(
s\right) $ and $\delta \vee \delta ^{\prime }\in \tilde{\sigma}^{-1}\left(
s^{\prime}\right) $.
\end{theorem}
In the special case where $s\rightarrow \tilde{\sigma}^{-1}\left( s\right) $ is a singleton, we recover the isotonicity of the inverse demand as in BGH.  Going further, by
taking $s=s^{\prime }$ in Theorem~\ref{thm:InverseIsotoneNoTies}, we obtain that, whenever it is non-empty, the set $\tilde{\sigma%
}^{-1}\left( s\right) $ is a \emph{lattice}\footnote{%
Whether the identified set is empty is considered in the Section \ref{sec:additional}.}:

\begin{corollary}
\label{cor:latticeStructureNoTies}Under Assumption~\ref{ass:Increasing} and~%
\ref{ass:NoIndiff}, if $\tilde{\sigma}^{-1}\left( s\right) $ is non-empty,
it is a lattice.  That is,
if $%
\delta, \delta^{\prime}\in \tilde{\sigma}^{-1}\left( s\right)$, then both $(\delta \wedge \delta ^{\prime
}), (\delta \vee \delta ^{\prime})\in \tilde{\sigma}^{-1}\left( s\right) $.
\label{cor:lattice}
\end{corollary}

This result is similar to~%
\citeasnoun{demange1985strategy}, who showed that the set of payoffs which
ensures a stable allocation is a lattice whenever it is non-empty.
This implies that the set of identified
utilities has a \textquotedblleft maximal\textquotedblright\ (resp.
\textquotedblleft minimal\textquotedblright ) element which is composed of
the \emph{component-wise} upper- (resp. lower-) bounds among all the utility
vectors in $\tilde\sigma^{-1}(s)$.   The upper bound corresponds to the unanimously most
preferred stable allocation for the consumers (\textquotedblleft
consumer-optimal\textquotedblright ) and the unanimously least preferred
stable allocation for the yogurts; 
conversely, the lower bound corresponds to the unanimously most preferred stable allocation for the
yogurts (\textquotedblleft yogurt-optimal\textquotedblright ) and least preferred for the consumers.  
%
Formally, define
\begin{equation*}
\tilde{\delta}_{j}^{\min }\left( s\right) =\min \left\{ \delta _{j}:\delta
\in \tilde{\sigma}^{-1}\left( s\right) \right\} \text{ and }\tilde{\delta}%
_{j}^{\max }\left( s\right) =\max \left\{ \delta _{j}:\delta \in \tilde{%
\sigma}^{-1}\left( s\right) \right\} .
\end{equation*}

Then the lattice property implies:
\begin{itemize}
\item[(i)] The set $\tilde{\sigma}^{-1}\left( s\right) $ has a minimal and a
maximal element:
\begin{equation*}
\tilde{\delta}^{\min }\left( s\right) \in \tilde{\sigma}^{-1}\left( s\right)
\text{ and }\tilde{\delta}^{\max }\left( s\right) \in \tilde{\sigma}%
^{-1}\left( s\right) .
\end{equation*}

\item[(ii)] Any $\delta \in \tilde{\sigma}^{-1}\left( s\right) $ is such that
  \begin{equation*}
\tilde{\delta}^{\min }\left( s\right) \leq \delta \leq \tilde{\delta}^{\max
}\left( s\right).
\end{equation*}

\item[(iii)] $\tilde{\sigma}^{-1}\left( s\right) $ is a singleton if and only
if
\begin{equation*}
\tilde{\delta}^{\min }\left( s\right) =\tilde{\delta}^{\max }\left( s\right)
.
\end{equation*}
\end{itemize}

Practically, most applications of partially
identified models focus on computing the component-wise upper and lower
bounds of the identified set of parameters; for general partially identified models, the vector of component-wise bounds will typically lie {\em outside} the (joint) identified set of parameters.   In contrast, our
lattice result here implies that these component-wise upper and lower bounds constitute {\em sharp} upper and lower bounds for the parameter vector as a whole, in the sense that they are attainable for selection mechanisms which place all probability on the highest (for upper bound) or lowest (for lower bound) payoffs for consumers.\footnote{In addition, the upper- and lower-lattice bounds here are the ``widest" possible in the sense that for each $j$, no utility lower than ${\delta^{\min}}_j$, or higher than  $\delta^{\max}_j$, can rationalize the observed set of market shares.   {Moreover, the equivalence theorem also implies that the identified utility set $\sigma^{-1}(s)$, while not necessarily convex, is connected, which derives from the properties of the core of matching games (section 9.2 of \citeasnoun{roth1992two}).} 
} 


In addition, the matching literature provides
algorithms to compute these extremal elements, which can be directly used to assess multiplicity: indeed, $\tilde{\sigma}^{-1}\left( s\right) $ is a single element (point-identified) if and only if its minimal and maximal
elements coincide. This flexibility in handling models for which the researcher may not know {\em a priori} whether model parameters are point- or partially-identified is an important contribution of the matching approach
developed in this paper.  We turn to these algorithms next.

\section{Matching-based Algorithms\label{sec:algorithms}}
The equivalence established in Theorem~\ref{thm:equivalenceNoTies} between
matching and discrete-choice models allows us to leverage several matching algorithms for both NARUMs and ARUMs. The use of these algorithms in the empirical discrete-choice literature is new.
In addition, matching-based algorithms have several advantages over existing procedures: (i) they can handle non-invertibility of the demand map, as all  of these algorithms allow for the case where multiple utility vectors can rationalize the same set of market shares; and (ii) these algorithms do not require smoothness of the demand map and therefore can handle some well-known models\textemdash like the pure characteristics model\textemdash which have non-smooth demand maps.  In contrast, existing algorithms for demand inversion  often rely on directly solving
the demand map $s_j=\sigma_j(\delta)$ for $\delta$ using fixed-point iterations or nonlinear-equation solvers, which typically requires smoothness of the demand map, and also rules out non-invertibility of the demand map.


We introduce three matching-based algorithms in this section.  The first is an algorithm for matching models with imperfectly transferable
utility (ITU)  which can be used for demand inversion in both  ARUMs or NARUMs.  This algorithm, called {\em Market Share Adjustment}, is essentially an ``accelerated'' version of the classic deferred acceptance algorithm (\citeasnoun{gale1962college}, \citeasnoun{crawford1981job}), which  iteratively adjusts the payoffs of the potential partners to achieve equilibrium. The second and third algorithms, in Section 4.2, are methods for computing stable allocations in two-sided matching models with transferable utility (TU), and apply only to ARUMs.  We consider a {\em linear-programming} approach based on the classic \citeasnoun{shapley1971assignment} assignment game, and a version of the {\em auction algorithm} of \citeasnoun{bertsekas1992auction}, augmented to produce bounds for partially-identified settings.

While the model in this paper  assumes a continuum of agents on each side of the market, for computational purposes we approximate this with a finite market populated by an equal (and large but finite) number, denoted $N$, of consumers and jars of yogurt.\footnote{In Appendix \ref{sec:semidiscrete} we present the {\em semi-discrete} algorithm for ARUMs, which is exact in that it computes the continuum problem directly.  However, as we explain there, its use  requires the unobserved taste vector to be (jointly) uniformly distributed over a polyhedron, which limits its general application.}    On the consumer side, each consumer $i\in \left\{1,...,N\right\} $ is characterized by a value of the utility shock $\varepsilon_i$ drawn i.i.d. (across $i$) from $P$, the distribution of the utility shocks. For a given vector of market shares $\left(s_0, s_1,\ldots, s_J\right)$, the number of jars of each brand $j$ of yogurt are set proportionately to the observed market share; that is,
$m_{j}\approx Ns_{j}\in \mathbb{N}$ of yogurts of type $j$ and, if needed, $m_{j}$ has
been rounded to an adjacent integer so that $\sum_{j\in \mathcal{J}_{0}}m_{j}=N$.
Throughout we maintain the utility normalization $\delta _{0}=0$.

\subsection{Deferred-acceptance algorithm (for both NARUMs and ARUMs)}

Theorem~\ref{thm:equivalenceNoTies} establishes an equivalence between the identified utility set and equilibrium payoffs in a two-sided matching game with imperfectly transferable utility.  Hence, for computing the identified utilities
one could use the deferred-acceptance algorithms developed in \citeasnoun{crawford1981job} and %
\citeasnoun{kelso1982job} which are generalizations of Gale and Shapley's \citeyear{gale1962college}
deferred-acceptance algorithm.  But these algorithms are very slow and inefficient,
especially in the common situation where there are fewer products (i.e. ``brands of yogurt'') than consumers. 

\subsubsection{Market shares adjusting algorithm (MSA)\label{par:MSA}}


As with all deferred-acceptance algorithms, there are two versions of the algorithm -- the ``consumer-proposing'' and ``yogurt-proposing'' versions --
return (resp.)  the lattice upper bound or lower bound on the utility parameters. In the case when the model is point identified, the upper bound and lower bound will coincide; hence running both versions of this algorithm yields a data-driven assessment of whether the model is point- or partially-identified.

In the ``consumer proposing'' version, the utilities ($\delta$'s) start at a high level, and consumers choose, in successive rounds, jars of yogurts which maximize their utilities. 
Between rounds, the utilities pertaining to the brands of yogurts in excess demand (ie. chosen by more consumers than available jars) are decreased by an adjustment factor.   Bidding continues until a round is
reached where the reference brand of yogurt ($j=0$) is in excess demand: that is, when the number of consumers choosing jars of brand 0 is greater or equal to the number of its available jars.
  In the original \citeasnoun{kelso1982job} version, the adjustment is done for {\em each jar} of yogurt separately, leading to very slow convergence with a large number of consumers and jars.
The MSA algorithm speeds this up by adjusting the utilities for {\em all jars of the same brand} of yogurt simultaneously.   Since this accelerated process can lead to \textquotedblleft
overshooting\textquotedblright\ (in which utilities move below their
equilibrium values), the MSA algorithm involves running a deferred-acceptance procedure multiple times with successively smaller adjustment factors.\footnote{
  This adjustment factor plays a role analogous to the step size in optimization procedures.   One typically chooses a larger step size in initial \emph{exploration} phases to move parameters away from regions where the optimum is unlikely to be. Choosing a larger step size speeds up the routine, but if the step size is too large, one might miss (``overshoot'') the optimum. Therefore, in later \emph{exploitation} phases, one decreases the step size to achieve a better accuracy.  Similar heuristics are used also to set the temperature parameter in simulated annealing, or the ``learning rate'' in machine learning procedures.
  }

We present below the pseudo code for the consumer-proposing version of the MSA algorithm, which obtains the lattice upper bound; Appendix \ref{algo for the lower bound} contains a version of the algorithm which yields the lattice lower bound.
Define $\bar{\delta}_{j}=\sup_{i\in \left\{ 1,...,N\right\} }\mathcal{U}%
_{\varepsilon_ij}^{-1}\left( \mathcal{U}_{i0}\left( \delta _{0}\right) \right) $.
Clearly, $\bar{\delta}$ is an upper bound for the stable payoffs and for the lattice upper bound.   Let $\eta^{tol}$ denote a small adjustment factor, which is a design parameter for the algorithm.

\begin{algorithm}[Consumer-proposing MSA]
\label{algo:msa}
\linespread{1}\small\normalsize
\par
\hfill \break
\noindent Define $\delta_{j}^{init}=\sup_{i\in \left\{ 1,...,N\right\} }\mathcal{U}%
_{\varepsilon_ij}^{-1}\left( \mathcal{U}_{i0}\left( \delta _{0}\right) \right) $
\Comment{\# Starting values above upper-bound}
\newline

\noindent\underline{Begin {\em\bfseries Adjustment (outer) Loop}}

Initialize $\eta ^{init}>>\eta^{tol}$ and $\delta _{j}^{init}=\delta^{return}_{j}$.

Repeat:

\qquad Call {\em\bfseries Deferred-Acceptance Loop} with $\left( \delta_{j}^{init},\eta ^{init}\right) $  which returns $\left( \delta', \eta'\right)$

\qquad Set $\delta ^{init}= \delta'+2\eta'$ and $%
\eta ^{init}= \eta'$

Until $\delta _{j}^{return}<\delta _{j}^{init}$ for all $j \in \mathcal{J}$. \Comment{\# Algo stops if all $\delta$ have decreased}

\noindent\underline{End {\em\bfseries Adjustment Loop}}

\bigskip

\noindent\underline{Begin {\em\bfseries Deferred-acceptance (inner) Loop}}

Require $\left( \delta _{j}^{init},\eta ^{init}\right) .$

Set $\eta =\eta ^{init}$ and $\delta=\delta ^{init}$.

While  $\eta \geq \eta ^{tol}$  \Comment{\# Run as long as tol. factor above threshold $\eta^{tol}$}

\qquad If $j\in \arg \max_{j}\mathcal{U}%
_{\varepsilon j}(\delta _{j})$  then $\pi _{ij}=1$ else $\pi _{ij}=0$  \Comment{\#  i is matched to optimizing j}

\qquad If $\sum\limits_{i}\pi _{i0} < m_{0}$, then for all $j\in \mathcal{J}$ \Comment{\# If brand 0 in excess supply then}

\qquad \qquad if $\sum\limits_{i}\pi _{ij}>m_{j}$ then $\delta _{j}= \delta _{j}-\eta$ . \Comment{\# decrease $\delta_j$, for brands $j$ w/ excess demand}

\qquad Else \Comment{\# Else if brand 0 in excess demand (overshooting)}

\qquad \qquad $\delta _{j}=
\delta _{j}+2\eta $ for all $j\in \mathcal{J}$ \Comment{\# Increase all $\delta$ (except $\delta_0$)}

\qquad \qquad $\eta = \eta /4.$ \Comment{\# Decrease the tol. factor for next loop}

End While

Return $\delta ^{return}=\delta $ and $\eta ^{return}=\eta $.

\noindent\underline{End {\em\bfseries Deferred-acceptance Loop}}

\end{algorithm}

As shown above, the consumer-proposing MSA consists of a deferred-acceptance loop nested inside of an adjustment loop.   In the deferred-acceptance loop, two cases can occur.  In the ``good case'', the deferred-acceptance loop starts with values $\delta_{j \in \mathcal{J}}$ above the lattice upper bound and an adjusting factor $\eta$ small enough, then all $\delta_{j \in \mathcal{J}}$ will reach the lattice upper bound without overshooting.  In the ``bad case'', that is when the deferred-acceptance loop starts with one $\delta_{j \in \mathcal{J}}$ below its upper bound and an adjusting factor $\eta$ small enough, then this $\delta_j$ will not decrease during the loop.
  The outer adjustment loop repeatedly calls the deferred-acceptance loop for decreasing values of the increment $\eta$.   It terminates when the utilities outputted by the deferred-acceptance loop have all decreased (indicating that the approximation loop is in the ``good case''); otherwise, the utilities are increased and the deferred-acceptance loop is called again.   

While we have not yet formally proven convergence of the MSA algorithm, we find that it runs remarkably quickly in all our simulations, relative to the \citeasnoun{crawford1981job} algorithm. In practice, one can assess the convergence of the algorithm in the outer loop by comparing the actual market shares with the predicted market shares evaluated using the $\delta$'s returned in each call to the inner loop.


\subsection{Transferable Utility Matching algorithms (for ARUMs)}

\label{sec:algo_arum}

For ARUMs, we can use algorithms for matching models with transferable utility (TU).   A
well-known result in optimal transport theory (see \citeasnoun{galichon2015optimal}) states that the equilibrium matching under
transferable utility maximizes the total surplus $\mathbb{E}\left[
\varepsilon _{\tilde{j}}\right] $ over all the distributions of $\left(
\varepsilon ,\tilde{j}\right) $ such that $\varepsilon \sim P$ and $\tilde{j}%
\sim s$, that is $\pi $ solves%
\begin{equation}\label{eq: OT_primal}
\max_{\left( \varepsilon ,\tilde{j}\right) \sim \pi \in \mathcal{M}\left(
P,s\right) }\mathbb{E}_{\pi }\left[ \varepsilon _{\tilde{j}}\right] .
\end{equation}%
The equilibrium payoffs $u$ and $\delta $ are the parameters which optimize the dual problem:
\begin{eqnarray}
\label{eq: OT_dual}
&&\inf\limits_{u,\delta :\delta _{0}=0}\left\{ \mathbb{E}_{P}\left[
u_{\varepsilon }\right] -\mathbb{E}_{s}\left[ \delta _{\tilde{j}}\right]
\right\}  \label{semi-discreteTransport} \\
s.t. &&u_{\varepsilon }-\delta _{j}\geq \varepsilon _{j}~\forall \varepsilon
\in \Omega ,~j\in \mathcal{J}_{0}.  \notag
\end{eqnarray}%

These are  convex programs which can be solved efficiently by linear programming (LP) or auction algorithms.  We contribute two novel modifications to the literature. For LP, we introduce a formulation that \emph{simultaneously} inverts multiple demand maps. Moreover, we augment both the LP and auction algorithms to compute lattice bounds for the identified utility set $\tilde\sigma^{-1}$. 

\subsubsection{Linear programming (Shapley-Shubik)\label{par:linprogr}}
For the large but finite discretization described above, the linear program (\ref{eq: OT_dual}) becomes \begin{eqnarray}
\label{eq: LP}
\inf\limits_{u_{i},\delta _{j}} &&\sum_{i=1}^{N}\frac{1}{N}u_{i}-\sum_{j\in
\mathcal{J}_{0}}s_{j}\delta _{j}  \\
s.t. &&u_{i}-\delta _{j}\geq \varepsilon _{j}^{i},\quad \forall i,j.\notag
\end{eqnarray}
which coincides with the dual of Shapley and Shubik's \citeyear{shapley1971assignment} classic assignment game.
While Shapley and Shubik showed that the set of optimizers in (\ref{eq: LP}) is a lattice, with bounds equal to the consumer-optimal and yogurt-optimal payoffs, they did not discuss how to compute these bounds.
  In principle, the upper (resp. lower) bounds can be obtained from the following problem:
  $$
\underset{u,\delta,\pi}{\max} \ \text{(resp. $\underset{u,\delta,\pi}{\min}$)}  \sum_{j=0}^J\delta_j\quad \text{s.t. $(u,\delta)\in \left\{\text{arginf}\ (\ref{eq: LP})\right\}$.}
  $$
This is a ``bilevel'' program, as the solutions to the LP in (\ref{eq: LP}) are used as the inputs into a second LP.\footnote{See \citeasnoun{Book_Bilevel_Program}.}   As is well-known, we can collapse a bilevel LP into a regular LP by replacing the lower-level LP (corresponding to (\ref{eq: LP})) with its optimality conditions.   That is, the upper (resp. lower) bounds can be obtained from the following LP:
\begin{align}
\begin{array}{rl}\label{eq: LP_bound_MPEC}
\underset{u,\delta,\pi}{\max} \ \text{(resp. $\underset{u,\delta,\pi}{\min}$)} & \sum_{j=0}^J\delta_j\\
\textrm{s.t.}& \sum_{i=1}^{N}\pi_{ij}=s_{j},\quad\forall j\\[10pt]
&\sum_{j=0}^{J}\pi_{ij}=\frac{1}{N},\quad\forall i\\[10pt]
&\pi_{ij}\geq 0,\quad\forall i, j\\[10pt]
&u_{i}-\delta_{j}\geq \epsilon_{ij},\quad\forall i, j\\[10pt]
&\sum_{i=1}^{N}\sum_{j=0}^{J}\pi_{ij}\epsilon_{ij}=\frac{1}{N}\sum_{i=1}^N u_{i}-\sum_{j=0}^J s_{j}\delta_{j},\\[10pt]
&\delta_{0}=0. \\
\end{array}
\end{align}
Constraints 1-3 in (\ref{eq: LP_bound_MPEC}) are the constraints of the primal assignment game (\ref{eq: OT_primal}), whereas constraint 4 appears in the dual problem (\ref{eq: OT_dual}).  Constraint 5 equates the primal and dual objectives at the optimum.  Taken together, these 5 constraints characterize the optimizing $(u,\delta)$ from (\ref{eq: OT_dual}).\footnote{See, for instance, \citeasnoun{Mangasarian1969}.}  Constraint 6 is our maintained normalization.

{\bfseries Combining LP problems.} A further benefit of the LP approach is that multiple demand inversion problems for different markets can be {\em combined} and solved simultaneously.   Specifically,
suppose there are $m=1,\dots,M$ markets. Instead of solving Problem (\ref{eq: LP}) for each of the $M$ markets separately, we combine these $M$ problems into \emph{one} problem to invert {\em all} demand maps \emph{simultaneously}:
\begin{eqnarray}
\label{eq: combined_LP}
\inf\limits_{u_{mi},\delta_{mj}} &&\sum_{m=1}^{M}\left(\sum_{i=1}^{N}\frac{1}{N}u_{mi}-\sum_{j\in
\mathcal{J}_{0}}s_{mj}\delta _{mj}\right)  \\
s.t. &&u_{mi}-\delta _{mj}\geq \varepsilon _{mj}^{i},\quad\forall i,j,m \notag
\end{eqnarray}

\noindent where all of the subscripts are augmented by the market index $m$. Since the decision variables $(u_{mi},\delta_{mj})$ and the associated constraints are market-specific, the resulting constraint coefficient matrix has a block-diagonal structure, which enables the use of efficient parallel sparse matrix routines for modern LP solvers.  We utilize this simultaneous demand inversion approach in the empirical application below, and confirm how one large but sparse problem (involving a larger number of parameters and constraints) is solved much more quickly than $M$ small problems.   

\subsubsection{Auction algorithms}\label{pgh:auctions}

Auction-type algorithms \`{a} la \citeasnoun{bertsekas1992auction} provide
an alternative approach to linear programming methods for solving TU-matching models. In these algorithms, unassigned
persons bid simultaneously for objects, decreasing their systematic
utilities (or equivalently raising their prices). Once all bids are in,
objects are assigned to the highest bidder. The procedure is iterated until
no one is unassigned. The description here follows 
\citeasnoun{bertsekas1989auction}.   We let $\kappa\in\left\{1,\ldots,N\right\}$ index jars of yogurt, where $j(\kappa)\in\mathcal{J}_0$ denotes the brand identity for the $\kappa$-th jar of yogurt.

We define the prices $p_{\kappa}$ as negative systematic utilities: $p_{\kappa} = -\delta_{\kappa}$.

\begin{algorithm}[Auction]
\label{algo:auction SO} Start with an empty assignment and a given vector of
prices $p_{\kappa}$ and set a scale parameter $\eta>0$.

\underline{Bidding phase}

(a) Each currently unassigned consumer $i$ chooses the jar $\kappa^{\star}$
to maximize utility:
\begin{equation}
\mathcal{U}_{\varepsilon_i \kappa^{\star}}(-p_{\kappa^{\star}}) =
\max_{\kappa} \mathcal{U}_{\varepsilon_i \kappa}(-p_{\kappa}).
\end{equation}

(b) Consumer $i$'s bid is set to:
\begin{equation}\label{eq:bidauction}
b_{i\kappa^{\star}} = p_{\kappa^{\star}} + \mathcal{U}_{\varepsilon_i
\kappa^{\star}}(-p_{\kappa^{\star}}) - w_i + \eta
\end{equation}
where $w_i$ denotes the utility from consumer $i$'s second-best choice:
\begin{equation}
w_i = \max_{j(\kappa) \neq j(\kappa^{\star})} \mathcal{U}_{\varepsilon_i
\kappa}(-p_{\kappa})
\end{equation}

\underline{Assignment phase}

Jar $\kappa$ is assigned to its highest bidder $i^{\star}$ and its price is raised to bidder $i^{\star}$'s bid:
\begin{equation}
p_{\kappa} : = b_{i^{\star}\kappa}
\end{equation}

\underline{Final step}

When no one is left unassigned the solution $\delta_j$ is recovered as:
\begin{equation}
\delta_j = - \min_{\kappa \in j(\kappa)} p_{\kappa}
\end{equation}
\end{algorithm}

Intuitively, the algorithm implements a Walrasian-style bidding procedure.
In each round, each unassigned
consumer bids for his favorite jar of yogurt. His bid (Eq. (\ref{eq:bidauction})) is equal to the difference
between the utilities from his most-preferred and second-most-preferred brands of yogurt (plus an extra $\eta>0$ factor to ensure that prices are increasing each round). The
consumer which makes the highest bid for a jar is assigned to it, and its
price is increased by the amount of the bid; a consumer previously assigned to this jar becomes unassigned and bids in the next round. The algorithms stops when all
individuals are assigned.
The performance of the algorithm is considerably improved by
applying it several times, starting with a large value of $\eta$ and
gradually decreasing it.

{\bfseries\em Computing utility bounds.}
The auction algorithm as described above works for the case when the demand map is invertible (so that the identified set of utilities is a singleton).   In practice, we do not know {\em a priori} whether the demand map is invertible or not; hence, we must extend the algorithm to allow for multiplicity of solutions.   To do this, we note that the auction algorithm described above 
returns the optimal allocation  $\pi_{ij}$  (for $i=1,\ldots, N$ and $j\in\mathcal{J}_0$) which is equal to one if consumer $i$ matches with a jar of yogurt brand $j$, and zero otherwise. 
In principle the utility bounds could be solved from the earlier linear programming problem (\ref{eq: LP_bound_MPEC}), with the matching indicators $\pi_{ij}$ fixed at the optimal allocation.
However, that would be inefficient as it doesn't fully exploit our knowledge of the optimal allocation.   Here we present a simpler alternative which is much quicker as it converges monotonically to the bounds of $\delta$.

{For the problem (\ref{eq: LP_bound_MPEC}), given the optimal allocation $\left\{ \pi_{ij}\right\}_{i,j}$, 
the solutions $(u,\delta)$ 
satisfy the inequalities 
$$u_{i}\geq \mathcal{U}_{\varepsilon _{i}j}(\delta _{j})\equiv \delta_j+\varepsilon_{ij}\ \forall i,j; \quad \text{
with equality for $\pi _{ij}>0$}.$$ 
 Accordingly, we construct an operator $T : \mathbb{R}^{\mathcal{I} \cup \mathcal{J} } \to\  \mathbb{R}^{\mathcal{I} \cup \mathcal{J} }$ given by
\begin{equation}\label{eq:bellmanford}
T_i(u,\delta) = \max(u_i, \max_{j\in\mathcal{J}_0} (\mathcal{U}_{\varepsilon_ij}(\delta_j))),\quad
T_j(u,\delta) = \max(\delta_j, \max_{i:\pi_{ij}>0} \mathcal{U}^{-1}_{\varepsilon_ij}(u_i)),\quad
\delta_0 =0.
\end{equation}
The set of fixed points of this operator  $u_i = T_i(u,\delta),\delta_j = T_j(u,\delta)$ satisfy the inequalities above. 
To see this, notice $\max (u_{i},\max_{j\in \mathcal{J}_{0}}(\mathcal{U}_{\varepsilon
_{i}j}(\delta _{j})))=u_{i}$  if and only if 
$u_{i}\geq \mathcal{U}_{\varepsilon _{i}j}(\delta _{j})$ for all $j\in \mathcal{J}_{0}$. Similarly, $\max (\delta _{j},\max_{i:\pi _{ij}>0} \mathcal{U}_{\varepsilon _{i}j}^{-1}(u_{i}))=\delta _{j}$ holds if and only if $\delta _{j}\geq \max_{i:\pi _{ij}>0}\mathcal{U}_{\varepsilon_{i}j}^{-1}(u_{i})$ $\Leftrightarrow$  $u_{i}\leq \mathcal{U}_{\varepsilon _{i}j}(\delta_{j})$) for all $i$ such that $\pi _{ij}>0$. Combining these, we obtain that the fixed point satisfies  
$u_{i}=\mathcal{U}_{\varepsilon _{i}j}(\delta
_{j})$ for all $i,j$ for which $\pi _{ij}>0$. }
%
Hence, by iterating on (\ref{eq:bellmanford}) beginning from the initial lowest possible values  $\left\{u_i = -\infty;\ \delta_j = -\infty,\ j \neq 0;\ \delta_0 = 0\right\}$, we will obtain a non-decreasing sequence of $\delta$'s converging to $\underline{\delta}$.

Analogously, starting with values of $+\infty $ and iterating on the
following operator we will obtain a monotonic non-increasing sequence of $\delta$'s converging to the upper bounds $\bar\delta$: 
\begin{equation}
\delta _{j}=\min (\delta _{j},\min_{i\in \mathcal{I}}\mathcal{U}%
_{\varepsilon _{i}j}^{-1}(u_{i})),\quad u_{i}=\min (u_{i},\min_{j:\pi
_{ij}>0}(\mathcal{U}_{\varepsilon _{i}j}(\delta _{j})),\quad \delta _{0}=0.
\label{eq:auctionupperbound}
\end{equation}
{ Indeed, $\delta _{j}=\min (\delta _{j},\min_{i\in \mathcal{I}}\mathcal{U}%
_{\varepsilon _{i}j}^{-1}(u_{i}))$ holds if and only if $\delta _{j}\leq
\min_{i\in \mathcal{I}}\mathcal{U}_{\varepsilon _{i}j}^{-1}(u_{i})$ $\Leftrightarrow$  $\mathcal{U}_{\varepsilon _{i}j}\left( \delta _{j}\right)
\leq u_{i}$, for all $i\in \mathcal{I}$. Similarly, $u_{i}=\min
(u_{i},\min_{j:\pi _{ij}>0}(\mathcal{U}_{\varepsilon _{i}j}(\delta _{j}))$
holds if and only if $u_{i}\leq \min_{j:\pi _{ij}>0}(\mathcal{U}%
_{\varepsilon _{i}j}(\delta _{j}))$ $\Leftrightarrow$  $u_{i}\leq 
\mathcal{U}_{\varepsilon _{i}j}(\delta _{j})$, for all $j$ such that $%
\pi _{ij}>0$.}\footnote{The iterations over the isotone operators (\ref{eq:bellmanford}) and (\ref{eq:auctionupperbound}) are instances of the {\em Bellman-Ford algorithm}, used for solving certain types of network flow problems, cf. \citeasnoun{sedgewick2011algorithms}.}

\subsection{Implementation}

We have developed \textsc{R} packages containing fast and efficient implementations of the algorithms described in this section. They are designed to enable user-friendly access for researchers to popular matching methods, and are used in the next section to benchmark the relative performance of each algorithm.
 \citeasnoun{github_auction} collects several auction and linear programming algorithms into an R package, utilizing \textsc{C++} code provided by \citeasnoun{walsh2017general}. Finally, a parallelized implementation of the MSA 
is contained in a separate package \citeasnoun{github_yogurts}. Links to the packages and installation instructions are provided in the bibliography references.

\section{Numerical experiments\label{sec:numerical experiments}}

We test our algorithms on two different models: the first one is the additive pure characteristics model, and the second one is a special case of a pure characteristics model which is {\em not} invertible, so that multiple values of utilities are consistent with a set of market shares.\footnote{While it is beyond the scope of this paper to consider formal statistical testing of equivalence between the bounds $\delta_{\min}=\delta_{max}$, in all our computational and empirical results below we computed the upper and lower bounds for different and increasing values of $N$, to ensure that any distance between $\delta_{min}$ and $\delta_{max}$ is not driven by approximation error due to small $N$.}

\subsection{The pure characteristics model}\label{sec: numerical pure char}
In this section we evaluate the performance of our matching-based algorithms in computing the pure characteristics model. \citeasnoun{berry2007pure} (p.
1193) underline that this model is appealing on theoretical grounds,\footnote{Models with logit errors have properties that may be undesirable for welfare analysis: They restrain substitution patterns and utility grows without bounds as the number of products in the market grows.  See \citeasnoun{berry2007pure} and \citeasnoun{ackerberg2005unobserved}.}   but it is infrequently used in empirical work, arguably due to the computational challenges associated with the non-smooth demand map.

We consider the following specification, which is adapted from \citeasnoun{dube2012improving}: The utility of consumer $i$ who chooses product $j$ in market $m$ is generated by:
\begin{equation}
    u_{mij}=\beta_0-\beta_p p_{mj}+\beta_1 x_{mj1}+\beta_2 x_{mj2}+\beta_3x_{mj3}+\xi_{mj},
\end{equation}
where $x_{mjk}, k=1,2,3$ are observed, exogenous product attributes, and $p_{mj}$ is the price of product $j$ in market $m$ that is correlated with the unobserved product attribute $\xi_{mj}$. {Instruments are constructed as noisy nonlinear functions of the exogenous $x$'s. The random coefficients are $(\beta_p,\beta_1,\beta_2,\beta_3)$,  distributed independently from normal distributions with means equal to $\bar\beta_p, \bar\beta_1,\bar\beta_2,\bar\beta_3$ and unit variances. Following \citeasnoun{berry2007pure}, the mean price coefficient $\bar\beta_p$ is normalized to -1. Additional details on the simulation setup and implementation are provided in Appendix \ref{sec:DGP_pure_char}.

We compare and contrast the performance of several algorithms.  Our algorithm, which we call  ``Matching-LP'', takes the form of a nested iterative procedure a la BLP (\citeasnoun{berry1995automobile}): that is, the procedure alternates between an outer and inner loop.   In the outer loop, a GMM objective function is optimized with respect to the model parameters, while the inner loop performs the demand inversion to recover the mean utilities ($\delta$) at the current candidate values of the model parameters.   In the inner loop, we utilize the linear programming (LP) algorithm from Section \ref{par:linprogr} to perform the demand inversion.} 

To benchmark performance, we compare our ``Matching-LP'' approach to two alternatives from the literature which reflect the current thinking on estimating the pure characteristics model. The first, denoted ``BLP-MPEC'', comes from \citeasnoun{berry2007pure}, who suggest  smoothing out the non-smooth pure characteristics demand map by adding small logit errors to each alternative, and then using the BLP estimation procedure.\footnote{In implementation, we utilize Dub\'{e}, Fox, and Su's (2012) MPEC algorithm,  which has  demonstrated speed advantages over nested methods. 
This speed advantage of MPEC arises in part from the use of Newton-type (gradient-based) iterations in the optimization procedure.  As \citeasnoun{LeeSeo2016} point out, Newton iterations have speed gains relative to BLP's original algorithm, which utilizes contraction mapping iterations.} {The second alternative, denoted ``PSL'', is the algorithm proposed in \citeasnoun{pang2015constructive} based on reformulating the model as a linear complementarity problem.} 

In Table \ref{table:pure_char_alg_overall}, we report the RMSE, bias, proportion of runs converged, and the average runtime across 20 Monte Carlo repetitions. There are 100 markets and 5 products (including the outside goods). In our simulations, we consider discretizations of the market into both  $N=500$ and $N=1000$ agents on each side of the market. We consider two model specifications: In Model I, we estimate the location parameters and fix all scale parameters. In Model II, we estimate the location parameter and the scale parameter associated with the endogenous price, fixing the rest of the scale parameters.

\begin{table}[!htbp] \centering
\caption{Numerical Performances: Estimating the Demand Parameters in Pure Characteristics Models}
\centering
\fontsize{7}{9}\selectfont
\begin{tabular}{llllllllrrrrr}
\toprule
\multicolumn{13}{c}{Panel I: Number of Consumers $N=500$}\\
\multicolumn{3}{c}{ } & \multicolumn{5}{c}{RMSE} & \multicolumn{5}{c}{Bias} \\
\cmidrule(l{3pt}r{3pt}){4-8} \cmidrule(l{3pt}r{3pt}){9-13}
algorithm & converge & runtime & $\beta_0$ & $\sigma_p$ & $\bar\beta_1$ & $\bar\beta_2$ & $\bar\beta_3$ & $\bar\beta_0$ & $\sigma_p$ & $\bar\beta_1$ & $\bar\beta_2$ & $\bar\beta_3$\\
& (\%)& (secs.)&&&&&&&&&&\\
\midrule
\colorbox{mygray}{Model I}&&&&&&&&&&&&\\
Matching-LP & 100 & 2.31 & 0.08 &  & 0.10 & 0.09 & 0.08 & -0.02 &  & 0.00 & -0.02 & 0.01\\
PSL & 85 & 134.42 & 0.08 &  & 0.11 & 0.10 & 0.08 & 0.00 &  & 0.00 & -0.02 & 0.01\\
BLP-MPEC & 100 & 32.88 & 0.19 &  & 0.17 & 0.18 & 0.14 & 0.09 &  & 0.01 & -0.02 & -0.04\\
&&&&&&&&&&&&\\
\hdashline
&&&&&&&&&&&&\\
\colorbox{mygray}{Model II}&&&&&&&&&&&&\\
Matching-LP & 100 & 76.52 & 0.09 & 0.06 & 0.11 & 0.11 & 0.07 & -0.02 & 0.01 & 0.00 & -0.02 & 0.00\\
PSL & 90 & 354.72 & 0.11 & 0.26 & 0.10 & 0.09 & 0.12 & -0.06 & -0.18 & 0.03 & 0.01 & 0.06\\
BLP-MPEC & 100 & 100.83 & 0.19 & 0.11 & 0.19 & 0.19 & 0.13 & 0.07 & -0.07 & 0.02 & 0.00 & -0.01\\
\midrule
\midrule
\multicolumn{13}{c}{Panel II: Number of Consumers $N=1000$}\\
\multicolumn{3}{c}{ } & \multicolumn{5}{c}{RMSE} & \multicolumn{5}{c}{Bias} \\
\cmidrule(l{3pt}r{3pt}){4-8} \cmidrule(l{3pt}r{3pt}){9-13}
algorithm & converge & runtime & $\beta_0$ & $\sigma_p$ & $\bar\beta_1$ & $\bar\beta_2$ & $\bar\beta_3$ & $\bar\beta_0$ & $\sigma_p$ & $\bar\beta_1$ & $\bar\beta_2$ & $\bar\beta_3$\\
& (\%)& (secs.)&&&&&&&&&&\\
\midrule
\colorbox{mygray}{Model I}&&&&&&&&&&&&\\
Matching-LP & 100 & 3.09 & 0.08 &  & 0.07 & 0.07 & 0.06 & -0.03 &  & -0.01 & -0.02 & 0.01\\
PSL & 90 & 409.65 & 0.08 &  & 0.07 & 0.07 & 0.06 & 0.01 &  & -0.01 & -0.02 & 0.01\\
BLP-MPEC & 100 & 79.38 & 0.27 &  & 0.16 & 0.17 & 0.16 & 0.14 &  & 0.02 & 0.01 & -0.09\\
&&&&&&&&&&&&\\
\hdashline
&&&&&&&&&&&&\\
\colorbox{mygray}{Model II}&&&&&&&&&&&&\\
Matching-LP & 100 & 120.19 & 0.09 & 0.07 & 0.07 & 0.07 & 0.07 & -0.02 & 0.00 & -0.01 & -0.01 & 0.01\\
PSL & 85 & 834.48 & 0.14 & 0.35 & 0.09 & 0.10 & 0.11 & -0.09 & -0.28 & 0.06 & 0.05 & 0.08\\
BLP-MPEC & 100 & 260.03 & 0.26 & 0.13 & 0.17 & 0.17 & 0.16 & 0.12 & -0.06 & 0.04 & 0.03 & -0.07\\
\bottomrule

\end{tabular}
\floatfoot{\linespread{1}\scriptsize{\emph{Note: }The numbers are averages across 20 Monte Carlo repetitions. There are 100 markets and 5 products.  The data-generating process is adapted from \citeasnoun{dube2012improving}, and complete details are contained in Appendix \ref{sec:DGP_pure_char}. $\beta_0$ denotes the constant in utility, while $\bar\beta_1,\bar\beta_2, \bar\beta_3$ correspond to the means of the random coefficients.  $\sigma_p$ is the scale parameter of price.  The true values are: $\beta_0=1,\ \bar\beta_1=0.5,\ \bar\beta_2=0.5,\ \bar\beta_3=0.2,\  \sigma=1.$ 
The inner loop of Matching-LP is based on the combined LP defined in Eq. (\ref{eq: combined_LP}).}}
\label{table:pure_char_alg_overall}
\end{table}

In Model I, Matching-LP and PSL have similar RMSE, which is roughly half of that of BLP-MPEC. In Model II, Matching-LP clearly dominates the other two alternatives: particularly, it delivers nearly an unbiased estimate for the standard deviation of the random coefficient of price ($\sigma_p$), while the other two approaches falter. 
In terms of computational speed, our method is on average 130 times faster than PSL and 26 times faster than BLP-MPEC in Model I. In model II, our method outperforms PSL by a factor of 7 and outperform BLP-MPEC by a factor of 2. These simulations demonstrate the superior performance of the matching-based approach for estimating the pure characteristics model.

While all three algorithms require random draws by simulating $N$ consumers from the distribution of the random coefficients, there is a fundamental difference in how these algorithms utilize the random draws. Both our matching-based approach and PSL use the random draws to \emph{discretize} the distribution, whereas the traditional BLP approach \emph{averages over} the random draws to approximate the market shares. Since simulation can potentially create bias, it is importantly to investigate the relationship between the simulation errors and the number of draws. By comparing Panel I and Panel II of Table \ref{table:pure_char_alg_overall}, we find that increasing the number of draws noticeably improves the RMSE of both our matching-based algorithm as well as PSL. For BLP-MPEC, however, there is little noticeable improvement.

As a supplementary exercise, we hone in on the performance of our matching algorithm in the demand inversion step itself, 
and compare the numerical accuracy of our matching-based algorithms to alternative existing approaches.  
Table \ref{table:pure_char_alg_inner} summarizes the numerical performance of three matching-based algorithms: (i) LP,
(ii) Auction, and (iii) MSA; the (iv) BLP contraction mapping (which adds logit errors to the utilities of each choice to smooth the demand map) is also included as a benchmark.\footnote{In our simulations we use the most common version of the BLP-contraction mapping algorithm, which utilizes iterations based on successive approximations.   \citeasnoun{LeeSeo2016} show that the performance of the algorithm can be improved by utilizing Newton iterations instead. However, both approaches require smoothness in the market-share mapping, which is not satisfied in the pure characteristics model considered here.}

\begin{table}[!htbp]
\centering
\begin{tabular}{l|cccc}
 Algorithms & Draws & Brands & RMSE & runtime (secs)\\ \toprule
 BLP contract. map. & 1,000 & 5 & 0.070 & 0.032  \\
  LP  & 1,000 & 5 & 0.029 & 0.083 \\
  Auction & 1,000 & 5 & 0.029 & 0.010 \\
  MSA & 1,000 & 5 & 0.029 & 18.776 \\
   \hline
BLP contract. map. & 1,000 & 50 & 0.045 & 0.283  \\
  LP  & 1,000 & 50 & 0.013 & 0.313 \\
  Auction & 1,000 & 50 & 0.013 & 0.046  \\
   \hline
BLP contract. map. & 1,000 & 500 & 0.018 & 2.774  \\
  LP  & 1,000 & 500 & 0.004 & 3.869  \\
  Auction & 1,000 & 500 & 0.005 & 0.426  \\
   \hline
BLP contract. map. & 10,000 & 5 & 0.072 & 0.331 \\
  LP  & 10,000 & 5 & 0.014 & 0.304  \\
  Auction & 10,000 & 5 & 0.014 & 0.117  \\
  MSA & 10,000 & 5 & 0.014 & 2.608  \\
   \hline
BLP contract. map. & 10,000 & 50 & 0.044 & 2.890  \\
  LP  & 10,000 & 50 & 0.006 & 4.446  \\
  Auction & 10,000 & 50 & 0.006 & 0.659  \\
   \hline
BLP contract. map. & 10,000 & 500 & 0.017 & 38.519 \\
  LP  & 10,000 & 500 & 0.002 & 64.061 \\
  Auction & 10,000 & 500 & 0.002 & 5.185  \\
 \hline
  \end{tabular}
\begin{tablenotes}
\footnotesize{
\item Note: The numbers are averaged from $50$ Monte-Carlo replications. Demand inversion for the pure characteristics model with $5$, $50$ and $500$ brands of yogurt and $1,000$ and $10,000$ draws of taste shocks.  The column "RMSE" corresponds to the root mean squared errors of the estimated $\delta_j$. {The MSA becomes computationally infeasible when the number of brands becomes large, so we report it only for scenarios with 5 brands.}}
\end{tablenotes}
\caption{Numerical Performances: Demand Inversion in Pure Characteristics Models}
\label{table:pure_char_alg_inner}
\end{table}

The numerical accuracy is almost identical across all three matching algorithms (LP, Auction, MSA). In comparison, the RMSE of the BLP contraction mapping is about three times larger. Similar to the finding in Table \ref{table:pure_char_alg_overall}, we find that increasing the number of draws $N$ improves the RMSE of our matching algorithms but not the BLP contraction mapping -- this arises from the additional approximation error due to the introduction of additive logit errors to smooth the mapping. 

For computational speed, the auction algorithm is the fastest by a wide margin: even under the most demanding scenario (500 brands and 10,000 simulated consumers), it only takes 5 seconds, which far outstrips other approaches.\footnote{{The computational speed reported in Table \ref{table:pure_char_alg_overall} and \ref{table:pure_char_alg_inner} are not directly comparable for several reasons: First, Table \ref{table:pure_char_alg_inner} only measures the computational time in the demand-inversion step\textemdash the ``inner loop'' in Table \ref{table:pure_char_alg_inner}. Second, in Table \ref{table:pure_char_alg_overall}, the differences between NFXP and MPEC for searching the structural parameters can contribute to the performance difference. Third, in Table \ref{table:pure_char_alg_overall}, there are 100 markets, whereas in Table \ref{table:pure_char_alg_inner}, there is only one market. The scalability of different demand-inversion methods also affects the overall computational speed.}} 

On the other hand, while MSA is the slowest algorithm, and does not scale up well with the number of brands, it is a general-purpose algorithm that also applies to NARUMs.\footnote{Interestingly, with 5 brands, MSA is faster with 10,000 draws than with 1,000: this may appear counter-intuitive, but increasing the number of draws improves accuracy during the loop and therefore reduces total computation times.} It is therefore not surprising that MSA is not as fast as the other methods since it does not exploit the additive separability in ARUMs. 

\subsection{A non-invertible pure characteristics model}\label{sec: numerical_non_invertible}

{For the second simulation exercise, we consider a pure characteristics model which is non-invertible\textemdash that is, there are multiple utility vectors which rationalizes the same set of market shares.  It underscores
an important contribution of our approach, namely its ability to
handle models which are not invertible.   Moreover, it also suggests that non-invertibility may be a typical feature of the pure characteristics model. \citeasnoun[pg. 654]{pang2015constructive} likewise note a problem of multiple solutions in  results from their estimation procedure for pure characteristics models.

There are three goods $y=1,2,3$, and the unknown parameters are the quality
of each good are $\delta_1,\delta_2,\delta_3$, with the normalization
$\delta_1=0< \delta_2,\delta_3$.
Consumers fall into two segments, which differ in the prices that they face.  In segment 1, prices are given by the vector $\mathbf{p}_1$, with $p^1_1<p^1_2<p^1_3$, while in segment 2, prices are given by the vector $\mathbf{p}_2$, with $p^2_1\leq p^2_3<p^2_2$.
The utility function is given by $$\mathcal{U}_{\varepsilon j}(\delta_j)=\delta_j - \frac{\varepsilon^b}{\varepsilon^a} p_j^1 +  \frac{(1-\varepsilon^b)}{\varepsilon^a} p_j^2,$$
which is a pure characteristics specification.   Consumer heterogeneity $\varepsilon\equiv (\varepsilon^a, \varepsilon^b)$ consists of both $\varepsilon^a\sim U[0,1]$, interpreted as willingness-to-pay for quality, and  $\varepsilon^b$, an indicator for whether they are in the first segment, which equals 0 or 1 with equal probability.}

  Let $s^d_j$ denote the
(unobserved) market share of good $j$ among segment $d (\in \left\{ 1,2\right\})$ customers. The observed market
shares are mixtures of market shares across both customer segments:
\begin{equation}  \label{1}
s_j=0.5(s^1_j+s^2_j),\quad j=1,2,3.
\end{equation}
This example typifies a common problem in supermarket scanner datasets, that consumers' transactions prices often differ markedly by products' list prices.\footnote{See \citeasnoun{erdem1998missing} for detailed discussions and modelling approaches to this problem.} 
In this case, the segment 1 prices $\mathbf{p^1}$ are analogous to "list" prices (ie., the prices posted on the supermarket shelves) while segment 2 customers have a coupon which gives them a discount on good 3. 
 
In the simulation exercise, we set $p^1_1=1$, $p^1_2=2$, and $p^1_3=3$ be the prices in segment 1, and $%
p^2_1=1$, $p^2_2=2$, and $p^2_3=1$ in segment 2. The observed market shares
are given by $(s_1,s_2,s_3)=(0.25,0.25,0.5)$. Given $%
\delta_1=0$, the identified utility set contains multiple values of the quality parameters $(\delta_2,\delta_3)$: 
\begin{equation}
\tilde\sigma^{-1}(s)=\{(\delta_2,\delta_3):\delta_2=2,\delta_3\in[1,3]\}
\end{equation}
which is a lattice  with minimal element $(0,2,1)$ and maximal element $(0,2,3)$.

Computational results using the LP algorithm are given in Table \ref{table:model2}. The algorithm
performs as expected; utilizing the linear programming algorithm (\ref{eq: LP}) along with the supplemental program to compute the upper and lower bound of payoffs (\ref{eq: LP_bound_MPEC}) accurately recovers the upper and lower bounds.  For comparison, we also performed the demand inversion exercise using the BLP contraction mapping approach, which requires the demand inverse to be unique.   Not surprisingly, we found that the contraction mapping algorithm would not converge for this problem.\footnote{On each trial run, the maximum number of function evaluations was reached without convergence.}

\begin{table}[h]
{
\centering
\resizebox{1\linewidth}{!}{
\begin{threeparttable}
\begin{tabular}{lc||cc|c||cc|c}
\hline
& True & \multicolumn{3}{c}{$N=100$} &
\multicolumn{3}{c}{$N=1000$} \\
& $\delta_j$ & LP - $\bar{\delta}_j$  & LP - $\underline{\delta}_j$ & BLP contr. map. & LP - $\bar{\delta}_j$  & LP - $\underline{\delta}_j$ & BLP contr. map. \\
\hline
$\delta_{2}$ & 2     & 2.106 \footnotesize{(0.305)} & 2.017 \footnotesize{(0.284)} & ---${}^{\dagger}$
& 2.002 \footnotesize{(0.071)}& 1.992 \footnotesize{(0.086)} & ---${}^{\dagger}$
\\
$\delta_{3}$ & [1,3] & 3.127 \footnotesize{(0.319)} & 0.997 \footnotesize{(0.286)} & ---${}^{\dagger}$
& 3.004 \footnotesize{(0.071)}& 0.990 \footnotesize{(0.086)} & ---${}^{\dagger}$
\\ \hline
\end{tabular}
\begin{tablenotes}
\linespread{1}\footnotesize{
\item We normalized $\delta_1=0$. The aggregate market shares are $(s_1,s_2,s_3)=(0.25,0.25,0.5)$, and prices: store 1 $(1,2,3)$, store 2 $(1,2,1)$. Lower and upper bounds are computed as the solution to the Linear Programming algorithm. For comparison, the BLP contraction mapping, an algorithm which requires the solution to be unique, failed to converge in each trial run. $N$ denotes the number of discretization points used. Standard deviations (across 50 replications) reported in parentheses.}
\item ${}^{\dagger}$: the BLP contraction mapping algorithm failed to converge.
\end{tablenotes}
\end{threeparttable}}
\caption{Non-invertible model: 3 goods}
\label{table:model2}
}
\end{table}

To continue, we consider an expanded version of this model involving 5
segments, and 8 goods. The prices and market shares for this example are given
in Table \ref{summary8}. Since, these market shares were chosen arbitrarily,
we do not know \emph{a priori} whether the identified utility set $\tilde\sigma^{-1}(s)$ is a singleton or contains multiple values.\footnote{%
That is, unlike the two-segment example used in Table \ref{table:model2}, we did not start by
computing a market equilibrium for given parameter values. Rather we chose
the prices and market shares in Table \ref{summary8} arbitrarily, and use our
approach to determine $\tilde\sigma^{-1}$.} 

In the two right-most columns of Table \ref{summary8}, we report the recovered upper and lower bounds for the $\delta$'s.
Clearly, for all entries, the lower and upper bounds coincide at the second or third decimal place, suggesting that the identified utility set $\tilde\sigma^{-1}$ is a singleton.  The results here were computed using $N=50000$, a very large value to ensure that the results are not driven by approximation error. In Table \ref{tab:approxN} in the appendix, we report results for different values of $N$, demonstrating that these results are robust and that, even with smaller values of $N$, the conclusions do not change.

\begin{table}[h]
{
\centering
\resizebox{\linewidth}{!}{
\begin{threeparttable}
  \begin{tabular}{l||ccccc|c||cc}
Brand & $p_1$ & $p_2$ & $p_3$ & $p_4$ & $p_5$ & MktShare & Lower-bound $\underline\delta$ & Upper-bound $\bar\delta$\\ \hline
A & 3.32 & 3.36 & 3.45 & 3.37 & 3.35 & 0.07 &  0.000 \footnotesize{(ref.)} &  0.000\footnotesize{(ref.)}   \\
B & 3.88 & 3.60 & 3.53 & 3.39 & 3.07 & 0.06 & -0.815 \footnotesize{(0.374)} & -0.805\footnotesize{(0.379)}   \\
C & 3.70 & 3.30 & 4.16 & 4.31 & 4.25 & 0.20 &  3.410 \footnotesize{(0.307)} &  3.414\footnotesize{(0.308)}   \\
D & 3.98 & 4.12 & 4.06 & 3.11 & 4.09 & 0.39 &  3.122 \footnotesize{(0.302)} &  3.125\footnotesize{(0.303)}   \\
E & 4.20 & 4.34 & 4.21 & 4.29 & 4.35 & 0.16 &  3.345 \footnotesize{(0.307)} &  3.349\footnotesize{(0.307)}   \\
F & 4.49 & 4.82 & 4.25 & 3.73 & 4.86 & 0.08 &  1.842 \footnotesize{(0.317)} & 1.848 \footnotesize{(0.317)}   \\
G & 7.13 & 7.92 & 7.95 & 7.99 & 7.71 & 0.01 &  6.798 \footnotesize{(0.312)} &  6.802\footnotesize{(0.312)}   \\
H & 8.34 & 8.37 & 8.59 & 8.62 & 8.67 & 0.05 &  8.018 \footnotesize{(0.313)} &  8.022\footnotesize{(0.313)}   \\ \hline
\end{tabular}
\begin{tablenotes}
\item \linespread{1}\footnotesize{
Estimates are LP solution computed with Gurobi 8.1. $50$ Monte-Carlo estimations were made with different draws of $\varepsilon^{a}$ (standard deviation across 50 replications reported  in parentheses). $50,000$ discretization points were used for each Monte-Carlo simulation ($i.e$ $50,000$ draws of $\varepsilon^{a}$ in the uniform). Brand A is the reference and the systematic utility is normalized to $0$.  
}
\end{tablenotes}

\end{threeparttable}
}
\caption{Multisegment price heterogeneity model: 5 segments and 8 goods}
\label{summary8}
}
\end{table}

\section{Empirical Application: Voting in European Parliament Elections \label{sec:empirical}}
Finally, we use our matching-based algorithm to estimate an aggregate spatial voting model using data from the 1999 Parliamentary Elections in the European Union countries, following \citeasnoun{MerloDePaula2017}.

\subsection{Model}
We consider a spatial voting framework in which both voters and political parties are characterized by their ``location'' in the political spectrum, which is the Cartesian plane $\mathbb{R}^2$. Voter $i$ is characterized by her ideal point $t_i\in \mathbb{R}^2$ within this space; likewise, candidate (political party) $j$ has an ideological position $C_j\equiv (C_{j1}, C_{j2})\in \mathbb{R}^2$. 

Voters are ideological: voter $i$ from electoral precinct $m$ votes for the party with the platform closest to her ideal point $t_{mi}$. Specifically, her preferred party is:
\begin{equation}
D_{mi}=\textrm{argmin}_{j\in J_{m}} d(t_{mi},C_{mj}),\label{eq: voter_DC}
\end{equation}
\noindent where $d(,)$: $\mathbb{R}^2\rightarrow \mathbb{R}_{+}$ is the distance function, which is specified as a quadratic function:
\begin{equation}\label{eq:quadratic}
d(t,C)=(t-C)^{'}W(t-C).
\end{equation}
$W$ is a $2\times 2$ weighting matrix, which for simplicity we assume to be identity: $W=\mathbb{I}_2$.  (This is restrictive as it implies that voters weigh both dimensions of the political spectrum equally.)
Unlike \citeasnoun{MerloDePaula2017} who consider the nonparametric identification and estimation of the distribution of voter ideal points $t$, we assume that they come from a bivariate normal distribution, as follows:  
\begin{align}
\begin{array}{l}
t_{mi1}=\left(X_m'\alpha + \nu_{mi1}\right),\\
t_{mi2}=\left(X_m'\beta + \nu_{mi2}\right),\label{eq: spec_latent_type}
\end{array} &\\
(\nu_{mi1},\nu_{mi2}) \sim N(\mathbf{0}, \Sigma%
),\ & \text{i.i.d. over $i$, $j$, $m$}
\end{align}
\noindent where $X_m=(x_{mk}$, $k=1,\dots,K)$, are aggregate demographic and economic variables in precinct $m$, and $\alpha$ and $\beta$ denote $K\times 1$ parameter vectors of interest.\footnote{As usual, we normalize the variance of $\nu_{mi1}$ to one. }  

By substituting (\ref{eq: spec_latent_type}) into (\ref{eq:quadratic}), the ``disutility'' that voter $i$ gets from voting for candidate $j$ becomes:
\begin{align}\nonumber
U_{mij} =\   & d(t_{mi}, C_{mj}) = (t_{mi1}-C_{mj1})^2 + (t_{mi2}-C_{mj2})^2 \\\nonumber
=\ &(X_m'\alpha)^2  + C_{mj1}^2 -  2\left(C_{mj1}*(X_m'\alpha)\right)+(X_m'\beta)^2  + C_{mj2}^2 -  2\left(C_{mj2}*(X_m'\beta)\right)\\
  &+\nu_{mi1}^2+2\left(X_m'\alpha-C_{mj1}\right)\nu_{mi1}+\nu_{mi2}^2+2\left(X_m'\beta-C_{mj2}\right)\nu_{mi2}\label{eq:purechar_distance0}
   \end{align}
   
 In the discrete-choice setting, only utility components which vary across candidates $j$ will affect choices.  From the above, 
 the terms $(\nu^2_{i1},\nu^2_{i2})$  and $(X_{m}'\alpha)\nu_{i1},(X_{m}'\beta)\nu_{i2}$ are individual-specific, and hence do not affect $i$'s choice problem. Therefore, our specification can be simplified into the following pure characteristics form:\footnote{\citeasnoun{MerloDePaula2017} also pointed out the equivalence of the spatial voting and pure characteristics models.} 
\begin{align}\label{eq:purechar_distance}
U_{mij} \equiv    \delta_{mj} +\varepsilon_{mij}& 
\\\nonumber
\text{where:}\quad\quad \delta_{mj} &=  C_{mj1}^2 + C_{mj2}^2-  2\left(C_{mj1}*(X_m'\alpha)\right)   -  2\left(C_{mj2}*(X_m'\beta)\right);\\\nonumber
\varepsilon_{mij}&=
-2C_{mj1}\nu_{mi1}-2C_{mj2}\nu_{mi2} 
\end{align}
The parameters in this specification are collectively denoted by $\theta\equiv \left( \alpha, \beta, \Sigma\right)$.

 As we pointed out above, the computational difficulties of the pure characteristics model has hampered empirical work utilizing it; indeed, despite the theoretical importance of the spatial voting model in political economy (ever since the pioneering work by \citeasnoun{downs1957economic}), empirical specifications have been relatively sparse.  Thus the application here, while primarily illustrative, does have a methodological contribution in showing how the matching-based algorithms introduced in this paper can be used to estimate a work-horse model from political science.

As data we use precinct-level vote shares from the 1999 European Parliament elections.
Altogether we have voting data from 822 electoral precincts in 22 regions, which are typically countries but may be sub-national regions distinguished by different sets of political parties.\footnote{The 22 regions are: Austria, Finland, France, Germany, Greece, Italy-Center, Italy-Islands, Italy-Northeast, Italy-Northwest, Italy-South, Portugal, Spain, Sweden, Netherlands, UK-East Midlands, UK-Eastern, UK-London, UK-Northwest, UK-Southeast, UK-Southwest, UK-West Midlands, and UK-Yorkshire.}
Parties' ideological positions are taken from \citeasnoun{HixNouryRoland2006}, who computed two ideological positions for each party: $C=(C_{j1}, C_{j2})$, with $C_{j1}$ denoting position on a left-right spectrum, and $C_{j2}$ denoting party's stance on the EU (with larger values denoting, resp. a more right-wing position and more pro-EU stance).   We use $K=3$ precinct-specific socio-economic and demographic variables: the female-to-male ratio, the proportion of the population older than 35 years, and the unemployment rate.\footnote{\citeasnoun{MerloDePaula2017} include additionally GDP as a regressor, and hence exclude Austria and Italy from their analysis due to missing values of this variable.}  All the data we use are available from the {\em Review of Economic Studies} website.\footnote{\url{https://doi.org/10.1093/restud/rdw046}}

As the quadratic spatial voting model is a pure characteristics model, we use our ``Matching-LP'' algorithm for estimation, which we also used for the simulations in Section 5.1.\footnote{We focus on the LP algorithm here in order to leverage combining and simultaneously solving the demand inversion problems across different precincts, which saves substantially on computational time.   Other algorithms, including the auction algorithm, are not amenable to combining the demand inversion problems.} 
Essentially, this resembles the \citeasnoun{berry1995automobile} estimation algorithm except that the inner loop utilizes a matching-based algorithm in place of the contraction mapping in BLP.  We describe the outer and inner loops in turn.


\subsubsection{Estimation: Inner loop}
In each call to the inner loop, for a given parameter vector $\theta$, we solve the demand inversion problems across all precincts simultaneously (see the earlier discussion associated with Eq. (\ref{eq: combined_LP})), by combining all the linear programs across all precincts into the following single large linear program:

\begin{eqnarray}\label{eq:combinedLP}
\inf\limits_{u_{mi},\delta _{mj}, m\in\mathcal{M}} &&\sum_{m=1}^{M}\left[\sum_{i=1}^{N}\frac{1}{N}u_{mi}-\sum_{j\in
\mathcal{J}_m}s_{mj}\delta _{mj}\right] \\
s.t. &&u_{mi}-\delta _{mj}\geq -\varepsilon_{mij}\quad\forall m\in\mathcal{M}
\end{eqnarray}
Let $\left\{\hat\delta_{mj}(\theta)\right\}_{m,j}$ denote the optimized values of the $\delta_{mj}$'s from this problem.  



\subsubsection{Estimation: Outer loop}
For the outer loop, we minimize the least-square differences between the $\hat\delta_{mj}(\theta)$'s emerging from the inner loop and the functional form for $\delta_{mj}$, as implied in Eq. (\ref{eq:purechar_distance}), to estimate $\theta$:
\begin{equation}\label{eq:outerloop}
\min_{\theta} \sum_{m,j} \left[ \hat\delta_{mj}(\theta) -  C_{mj1}^2 - C_{mj2}^2-  2C_{mj1}*(X_m'\alpha)   -  2C_{mj2}*(X_m'\beta)\right]^2
\end{equation}

\subsection{Results}
\begin{table}[!htbp] \centering
  \caption{Estimation Results: 1999 European Parliamentary Elections}
{
\begin{tabular}{@{\extracolsep{5pt}}lcc}
\toprule
&Model I    &Model II\\
\midrule
\colorbox{mygray}{dim. 1: right-leaning}&&\\
                                        &               & \\
 female-to-male ratio                   &    0.81     &      0.44 \\
                                        &    (0.50,1.12) &    (0.19,0.68)  \\
 above 35-year-old($\%$)                &   -0.07     &     -0.29  \\
                                        &  (-0.56,0.49) &   (-0.68,0.12)  \\
 unemployment rate                      &   -2.62     &     -1.76  \\
                                        & (-3.02,-2.17) &  (-2.16,-1.33)  \\
\hdashline
\colorbox{mygray}{dim. 2: pro-EU}       &               &  \\
                                        &               &  \\
 female-to-male ratio                   &   0.62      &     -0.48 \\
                                        &  (0.29,0.96)  &  (-0.85,-0.10)  \\
 above 35-year-old($\%$)                &  -0.02  &     0.80 \\
                                        &  (-0.60,0.55)  &   (0.16,1.39)  \\
 unemployment rate                      &  -2.95      &    -1.25 \\
                                        &(-3.53,-2.28)  & (-1.96,-0.62)  \\
\hdashline
$\rho$                                  & 0               & -0.66   \\
                                        & (fixed)              &(-0.70,-0.40)    \\
{\em Checking multiplicity:} $\max_{m,j}\{\delta^{UB}_{mj}-\delta^{LB}_{mj}\}$ & $1.52 \times 10^{-9}$&$9\times 10^{-4}$   \\

\bottomrule
\end{tabular}
}
\floatfoot{\emph{Note: }\scriptsize{
    The inner loop utilizes the linear programming algorithm (Section 4.2.1). We report the 95\% bootstrap confidence interval constructed from 500 bootstrap samples. We use $N=100$ voters in each precinct. $\delta_{mj}^{UB}$ and $\delta_{mj}^{LB}$ refer, respectively, to estimates of the upper and lower bound of the identified utility parameters, computed using the algorithm in Eq. (\ref{eq: LP_bound_MPEC}).\\
}
}\label{tab:empirical}
\end{table}

Table \ref{tab:empirical} contains our estimation results for two specifications.\footnote{For this ARUM, without loss of generality, we normalized $\delta_{m1}=0$, corresponding to the party listed first in alphabetical order in each precinct $m$.  We implemented this normalization by subtracting, in each market, the covariates $C_{j1}^2 + C_{j2}^2-  2\left(C_{j1}*(X_m'\alpha)\right)   -  2\left(C_{j2}*(X_m'\beta)\right)$ for candidate 1 from the utilities for the other candidates $j\neq 1$ (cf. Eq. (\ref{eq:purechar_distance})).  The pure characteristics model also requires a location normalization (cf.\citeasnoun[footnote 13]{berry2007pure}) which is automatically satisfied as the spatial voting model implies that the constant term in the utility is equal to the squared ideological positions $C_{j1}^2+C_{j2}^2$ which is known by the researcher.} In Model I, we assume that the covariance matrix  $\Sigma=\mathbb{I}_2$, thus ruling out correlation between the two ideal point dimensions conditional on precinct characteristics $X_m$.  Model II eliminates this restriction and allows for correlation between the dimensions, so that $\Sigma=
\begin{pmatrix}
1&\rho\\
\rho&1
\end{pmatrix}
$.
The estimate for $\rho$ in Model II is negative (coef. -0.66) and significantly different from zero, implying that, after controlling for precinct-level characteristics, voters who are right-leaning (dimension 1) tend to be {\em less} in favorable towards the EU (dimension 2).\footnote{In comparison, the richer model in \citeasnoun{MerloDePaula2017} allows the correlation between voter dimensions to differ by country, and they estimate this correlation to be negative in 4 out of the 10 countries, including France and Germany, the two largest EU countries (cf. Table 3 in \citeasnoun{MerloDePaula2017}).} Given this result, we focus the discussion here to the results from Model II.

The estimated coefficients summarize the contribution of demographic variables on the two dimensions of voters' ideal points.   For the first dimension, the coefficient on unemployment rate is negative (-1.76), indicating that left-leaning precincts tend to have higher unemployment rates.  Precincts with larger female-to-male ratio tend to be more right-leaning (0.44), while there is no significant relation between age (measured by the proportion of population older than 35 years) on the tendency to be pro-conservative.   

For the second dimension, we find that precincts with higher unemployment rate are significantly less supportive of the EU (coef. -1.25).
Precincts with higher percentage of older voters are significantly more pro-EU, while those with a higher female-male ratio are less supportive of the EU. 
Altogether, economic considerations -- as exemplified in the unemployment rate -- appear to be the strongest and most consistent explanator of voters' preferences across European regions. 

We also use our algorithms to check whether the demand map is invertible by solving for the upper and lower bounds on the precinct/party qualities $\delta_{mj}$ using the linear program in Eq. (\ref{eq: LP_bound_MPEC}).  At the bottom of Table \ref{tab:empirical}, we report the maximal (across all precincts and parties) difference between the estimated upper and lower bounds in the $\delta_{mj}$'s, evaluated at the parameter values reported in Table~\ref{tab:empirical}.  As is evident, the difference is minuscule, suggesting that multiplicity of these parameters is not an issue for this model.\footnote{
  With multiple $\delta$'s, the identification and estimation of the structural parameters (the $\beta$'s and the parameters in the distribution of random coefficients) is an open question, and we do not consider it here.  Estimating these parameters typically relies upon instruments, and the associated moment conditions. The literature on identification and inference in moment condition models with possibly partially-identified parameters is still nascent; see \citeasnoun{chen2018monte} for one recent paper.  In ongoing work (\citeasnoun{HsiehMonteroShum}) we are tackling this issue in the context of discrete-choice demand models as here, but both the issue regarding identification of $\beta$, as well as the computational and inferential issues associated with estimating $\beta$ in this setting are challenging.
  }
Finally, the combined linear program in Eq. (\ref{eq:combinedLP}) is a big time-saver, as executing the demand inversion problem {\em simultaneously} across all precincts is ten times faster than performing the demand inversion separately for each precinct.  This suggests that simultaneous solution of demand inversion problems constitutes a very practical advantage of the linear programming approach.

\section{Concluding remarks}\label{sec: conclusion}

In this paper we have explored the intimate connection between discrete
choice models and two-sided matching models, and used results from the
literature on matching under imperfectly transferable utility to derive
procedures for demand inversion and estimation of discrete choice models based on
the non-additive random utility specification.
Although the microeconomics literature distinguishes
between \textquotedblleft one-sided\textquotedblright\ and \textquotedblleft
two-sided\textquotedblright\ demand problems, our results show that this
distinction is immaterial for the purpose of estimating discrete-choice models. Given the matching equivalence, it is
as appropriate to consider the discrete choice problem of
consumers choosing yogurts as one in which
\textquotedblleft yogurts choose consumers\textquotedblright .

The connection between discrete choice and two-sided matching is a rich one,
and we are exploring additional implications. For instance, the phenomenon
of \textquotedblleft multiple discrete choice\textquotedblright\ (consumers
who choose more than one brand, or choose bundles of products on a purchase
occasion) is challenging and difficult to model in the discrete choice
framework\footnote{%
See \citeasnoun{hendel1999estimating}, \citeasnoun{dube2004multiple}, %
\citeasnoun{fox2013measuring} for some applications.} but is quite natural
in the matching context, where \textquotedblleft
one-to-many\textquotedblright\ markets are commonplace -- perhaps the most
prominent and well-studied being the National Residents Matching Program for
aspiring doctors in the United States (cf. \citeasnoun{roth1984evolution}). We are exploring
this connection in ongoing work.


\bibliographystyle{ECMA}
\bibliography{mybib}
\newpage \appendix

{\bfseries (Intended for online publication)}

\linespread{1.2}

\section{Additional theoretical results}\label{sec:additional}
Here we derive additional theoretical results for the matching model in Section \ref{par:noIndiff}.

\subsection{Existence\label{par:existence}}
As our main theoretical results explore a new equivalence between the identified utility set $\tilde\sigma^{-1}$ and the equilibrium payoffs in a two-sided matching game,  we start by considering the existence of the equilibrium payoff set in this game under our assumptions.\footnote{
We note that this result is a contribution
to matching theory per se, as it implies the existence of a solution to
the equilibrium transport problem, as introduced in~%
\citeasnoun{galichon2015optimal}, Definition 10.1. }
In order to show that $\tilde{\sigma}^{-1}\left( s\right) $ is non-empty, we
need to make slightly stronger assumptions than the ones that were
previously imposed. In particular, Assumption~\ref{ass:Increasing} will be
replaced by the following:

\begin{assumption}[Stronger regularity of $\mathcal{U}$]
\label{ass:ContinuousMeasurable} Assume:

(a) for every $\varepsilon \in \Omega $, the map\ $\varepsilon \mapsto
\mathcal{U}_{\varepsilon j}\left( \delta _{j}\right) $ is integrable, and

(b) the random map $\delta _{j}\mapsto \mathcal{U}_{\varepsilon j}\left(
\delta _{j}\right) $ is stochastically equicontinuous.
\end{assumption}

We also need to keep track of the behavior of $\mathcal{U}_{\varepsilon
j}\left( \delta \right) $ when $\delta $ tends to $-\infty $ or $+\infty $,
and for this, we introduce the following assumption:

\begin{assumption}[Left and right behavior]
\label{ass:leftRightBehav}Assume that:

(a) There is $a>0$ such that $\mathcal{U}_{\varepsilon j}\left( \delta \right) $
converges in probability as $\delta \rightarrow -\infty $ towards a random
variable dominated by $-a$, that is: for all $\eta >0$, there is $\delta
^{\ast }\in \mathbb{R}$ such that $\Pr \left( \mathcal{U}_{\varepsilon j}\left(
\delta ^{\ast }\right) >-a\right) <\eta $.

(b) $\mathcal{U}_{\varepsilon j}\left( \delta \right) $ converges in probability
as $\delta \rightarrow +\infty $ towards $+\infty $, that is: for all $\eta
>0$ and $b\in \mathbb{R}$, there is $\delta ^{\ast }\in \mathbb{R}$ such
that $\Pr \left( \mathcal{U}_{\varepsilon j}\left( \delta ^{\ast }\right)
<b\right) <\eta $.
\end{assumption}

We define $\mathcal{S}_{0}^{int}=\left\{ s\in \mathcal{S}_{0}:s_{j}>0,~%
\forall j\in \mathcal{J}_{0}\right\}$. We can now prove the existence
theorem.

\begin{theorem}
\label{thm:existenceNoTies}Under Assumptions~\ref{ass:Increasing},~\ref%
{ass:NoIndiff},~\ref{ass:ContinuousMeasurable}, and~\ref{ass:leftRightBehav}%
, $\tilde{\sigma}^{-1}\left( s\right) $ is non-empty for all $s\in \mathcal{S%
}_{0}^{int}$.
\end{theorem}

\subsection{Uniqueness and convergence\label{par:uniqueness}}

Next we consider uniqueness. Assume the random maps $\delta \mapsto
\mathcal{U}_{\varepsilon j}\left( \delta \right) $ are invertible for each $%
\varepsilon \in \Omega $ and $j\in \mathcal{J}$, and define $Z_0$ to be the
random vector such that $Z_{j0}=\mathcal{U}_{\varepsilon j}^{-1}\left(
\mathcal{U}_{\varepsilon 0}\left( \delta _{0}\right) \right) $. $Z$ is a
random vector valued in $\mathbb{R}^{\mathcal{J}}$; let $P_{Z_0}$ be the
probability distribution of $Z_0$.
We will consider the following assumption on $P_{Z_0}$.

\begin{assumption}
\label{ass:noHole}Assume that:

(i) the map $\delta _{j}\mapsto \mathcal{U}_{\varepsilon j}\left( \delta
_{j}\right) $ is invertible for each $\varepsilon \in \Omega $ and $j\in
\mathcal{J}$, and

(ii) $P_{Z_0}$ has a non-vanishing density over $\mathbb{R}^{J}$.
\end{assumption}

Given these assumptions, the following result is a direct consequence of BGH's result on the invertibility of demand systems
\begin{theorem}
\label{thm:uniquenessNoTies}Under Assumptions~\ref{ass:Increasing},~\ref%
{ass:NoIndiff} and~\ref{ass:noHole}, $\tilde{\sigma}^{-1}\left( s\right) $
has a single element for all $s\in \mathcal{S}_{0}^{int}$.
\end{theorem}

Assumption~\ref{ass:noHole} is quite natural. In the case of
additive random utility models, the map $\delta _{j}\mapsto \mathcal{U}%
_{\varepsilon j}\left( \delta _{j}\right) =\delta _{j}+\varepsilon _{j}$ is
indeed continuous, \ and $Z_{j}=\delta _{0}+\varepsilon _{0}-\varepsilon
_{j} $ has a non-vanishing density over $\mathbb{R}^{J}$ when $\left(
\varepsilon _{0}-\varepsilon _{j}\right) $ does.   On the other hand, for the model in Section 5.2, Assumption 5(ii) is violated, because the consumer heterogeneity $\varepsilon^a\in [0,1]$, the random vector $Z$ can only have bounded support.

Given uniqueness, we also consider convergence properties.
While we do not focus on statistical inference in this paper, the next result may be useful for showing asymptotic properties of our procedures.
In practice, the vector of market shares $s$ may contain sample
uncertainty, and we may approximate $P$ by discretization. This will
provide us with a sequence $\left( P^{n},s^{n}\right) $ which converges
weakly toward $\left( P,s\right) $, where $P$ is the true distribution of $%
\varepsilon $, and $s$ is the vector of market shares in the population.
Under assumptions slightly weaker than for Theorem~\ref{thm:uniquenessNoTies}, we
establish that if $P^{n}$ and $s^{n}$ converge weakly to $P$ and $s$,
respectively, then any $\delta ^{n}\in \tilde{\sigma}^{-1}\left(
P^{n},s^{n}\right) $ will also converge.

\begin{assumption}
\label{ass:consistency}Assume that:

(i) the map $\delta _{j}\mapsto \mathcal{U}_{\varepsilon j}\left( \delta
_{j}\right) $ is invertible for each $\varepsilon \in \Omega $ and $j\in
\mathcal{J}$, and

(ii) for each $\delta \in \mathbb{R}^{\mathcal{J}}$, the random vector $%
\left( U_{\varepsilon j}\left( \delta _{j}\right) \right) _{j\in \mathcal{J}%
} $ where $\varepsilon \sim P$ has a non-vanishing continuous density $%
g\left( u;\delta \right) $ such that $g:\mathbb{R}^{\mathcal{J}}\times
\mathbb{R}^{\mathcal{J}}\rightarrow \mathbb{R}$ is continuous.
\end{assumption}

Note that Assumption~\ref{ass:consistency} is stronger than Assumption~\ref%
{ass:noHole}.

\begin{theorem}
\label{thm:consistency}Under Assumptions~\ref{ass:Increasing},~\ref%
{ass:NoIndiff} and~\ref{ass:consistency}, assume that $P^{n}$ and $s^{n}$
converge weakly to $P$ and $s$, respectively. By theorem~\ref%
{thm:uniquenessNoTies}, $\tilde{\sigma}^{-1}\left( P,s\right) $ is a
singleton, denoted $\left\{ \delta \right\} $. If $\delta ^{n}\in \tilde{%
\sigma}^{-1}\left( P^{n},s^{n}\right) $ for all $n$, then $\delta
^{n}\rightarrow \delta $ as $n\to \infty $.
\end{theorem}

\section{\protect\small Proofs}

\label{app:proofs}

\subsection{{\protect\small Proof of theorem~\protect\ref%
{thm:equivalenceNoTies}}}

\begin{varproof}
{\small \underline{(a) From demand inversion to equilibrium matching}:
Consider $\delta \in \tilde{\sigma}^{-1}\left( s\right) $ a solution to the
demand inversion problem. Then $s_{j}=P\left( \varepsilon \in \Omega :%
\mathcal{U}_{\varepsilon j}\left( \delta _{j}\right) \geq u\left(
\varepsilon \right) \right) $, where
\begin{equation}
u\left( \varepsilon \right) =\max_{j\in \mathcal{J}_{0}}\mathcal{U}%
_{\varepsilon j}\left( \delta _{j}\right) .\label{max-consumer-surplus}
\end{equation}%
Let us show that we can construct $\pi $ and set $v=-\delta $ such that $%
\left( \pi ,u,v\right) $ is an equilibrium outcome, which is to say it
satisfies the three conditions of Definition~\ref{def:equilibrium}. 

Let us show that the solution of problem~(\ref{max-consumer-surplus}) is attained with probability one by one element denoted $j(\varepsilon)$. Indeed, otherwise, $\left\{
j,j^{\prime }\right\} \subseteq J\left( \varepsilon \right) $ would arise
with positive probability for some pair $j\neq j^{\prime }$, which would
imply that there is a positive probability of indifference between $j$ and $%
j^{\prime }$, in contradiction with~(\ref{defMktShare}). Hence we have $ \{j\left( \varepsilon \right) \} =\arg \max_{j\in
\mathcal{J}_{0}}\left\{ \mathcal{U}_{\varepsilon j}\left( \delta _{j}\right)
\right\} $ with probability one, and the random variable $j\left( \varepsilon \right)$ has probability distribution  $s$. 

Introduce $\pi$ as the joint distribution of $\left( \varepsilon,j\left( \varepsilon \right) \right)$. We have by definition that $\pi $ has marginal distributions $P$ and $s$: 
$$\pi \in \mathcal{M}%
\left( P,s\right). $$

Next, introducing $v\left( j\right) =-\delta _{j}$, $g_{\varepsilon
j}\left( v\left( j\right) \right) =-\mathcal{U}_{\varepsilon _{i}j}\left(
\delta _{j}\right) $, and $f_{\varepsilon j}\left( x\right) =x$, we see that expression~(\ref{max-consumer-surplus}) implies that for all $\left( \varepsilon ,j\right) $,
\begin{equation}
f_{\varepsilon j}\left( u\left( \varepsilon \right) \right) +g_{\varepsilon
j}\left( v\left( j\right) \right) \geq 0.\label{stabil-again}
\end{equation}%
However, for $(\varepsilon,j)$ in the support of $\pi $ we have that $j$ is optimal in~(\ref{max-consumer-surplus}), and therefore equality in~(\ref{stabil-again}) holds. 

Hence, conditions (i), (ii) and (iii) in Definition~\ref{def:equilibrium}
are met, and $\left( \pi ,u,v\right) $ is an equilibrium outcome. 

\underline{(b) From equilibrium matching to demand inversion}: Let $%
\left( \pi ,u,v\right) $ be an equilibrium matching in the sense of
Definition~\ref{def:equilibrium}, where $f_{\varepsilon j}\left( x\right) =x$
and $g_{\varepsilon j}\left( y\right) =-\mathcal{U}_{\varepsilon
_{i}j}\left( -y\right) $. Then letting $\delta =-v$, one has by condition
(ii) that for any $\varepsilon \in \Omega $ and $j\in \mathcal{J}_{0}$,
\begin{equation*}
u\left( \varepsilon \right) -\mathcal{U}_{\varepsilon j}\left( \delta
_{j}\right) \geq 0
\end{equation*}%
thus $u\left( \varepsilon \right) \geq \max_{j\in \mathcal{J}_{0}}\mathcal{U}%
_{\varepsilon j}\left( \delta _{j}\right) $. But by condition (iii), for $%
j\in Supp\left( \pi \left( .|\varepsilon \right) \right) $, one has $u\left(
\varepsilon \right) =\mathcal{U}_{\varepsilon j}\left( \delta _{j}\right) $,
thus%
\begin{equation*}
u\left( \varepsilon \right) =\max_{j\in \mathcal{J}_{0}}\mathcal{U}%
_{\varepsilon j}\left( \delta _{j}\right) .
\end{equation*}%
Condition (iii) implies that if $\left( \tilde{\varepsilon},\tilde{j}\right)
\sim \pi $, then $\Pr \left( j\left( \tilde{\varepsilon}\right) =
\tilde{j} \right) =1$, thus
\begin{equation*}
\sigma _{j}\left( \delta \right) =P\left( \varepsilon \in \Omega :j\left(
\tilde{\varepsilon}\right) = j \right) =\Pr \left( \tilde{j}%
=j\right) =s_{j}.
\end{equation*}%
Hence $s = \tilde{\sigma}\left( \delta \right) $.]\hfill QED }
\end{varproof}

\subsection{{\protect\small Proof of theorem~\protect\ref%
{thm:InverseIsotoneNoTies} }}

\begin{proof}

{\small

Assume $s_{j}\leq s_{j}^{\prime }$ for all $j\in \mathcal{J}$ (which implies 
$s_{0}\geq s_{0}^{\prime }$).

Let $s^{\wedge }=\sigma \left( \delta \wedge \delta ^{\prime }\right) $, and
show that $s^{\wedge }=s$. 

We have for all $j\in \mathcal{J}_{0}$, 
\[
s_{j}^{\wedge }=\Pr \left( j\in \arg \max_{j\in \mathcal{J}_{0}}\mathcal{U}%
_{\varepsilon j}\left( \delta _{j}\wedge \delta _{j}^{\prime }\right)
\right) 
\]

So for $j\in \mathcal{J}_{0}$ with $\delta _{j}\leq \delta _{j}^{\prime }$.
Then $\delta _{j}\wedge \delta _{j}^{\prime }=\delta _{j}$ and $\delta
_{k}\wedge \delta _{k}^{\prime }\leq \delta _{k}$ for $k\neq j$, hence when
moving from $\delta $ to $\delta \wedge \delta ^{\prime }$, the utility
associated with $j$ has remained unchanged, while the utilities associated
with alternatives have weakly decreased, and hence in that case, $%
s_{j}^{\wedge }\geq s_{j}$.

Now assume $\delta _{j}>\delta _{j}^{\prime }$. Then the same logic implies $%
s_{j}^{\wedge }\geq s_{j}^{\prime }$, but because $s_{j}^{\prime }\geq s_{j}$
as $j\neq 0$, we have also $s_{j}^{\wedge }\geq s_{j}$.

As a result $s_{j}^{\wedge }\geq s_{j}$ for all $j\in \mathcal{J}_{0}$. But
as these two vectors both sum to one, we in fact get equality.

\bigskip 

Now let $s^{\vee }=\sigma \left( \delta \vee \delta ^{\prime }\right) $, and
show that $s^{\vee }=s^{\prime }$. \ We have for all $j\in \mathcal{J}_{0}$, 
\[
s_{j}^{\vee }=\Pr \left( j\in \arg \max_{j\in \mathcal{J}_{0}}\mathcal{U}%
_{\varepsilon j}\left( \delta _{j}\vee \delta _{j}^{\prime }\right) \right) 
\]

So for $j\in \mathcal{J}_{0}$ with $\delta _{j}\leq \delta _{j}^{\prime }$.
Then $\delta _{j}\vee \delta _{j}^{\prime }=\delta _{j}^{\prime }$ and $%
\delta _{k}\vee \delta _{k}^{\prime }\leq \delta _{k}^{\prime }$ for $k\neq j
$, hence when moving from $\delta $ to $\delta \vee \delta ^{\prime }$, the
utility associated with $j$ has remained unchanged, while the utilities
associated with alternatives have weakly increased, and hence in that case, $%
s_{j}^{\vee }\leq s_{j}$.
Now assume $\delta _{j}>\delta _{j}^{\prime }$. Then the same logic implies $%
s_{j}^{\vee }\leq s_{j}$, but because $s_{j}\leq s_{j}^{\prime }$ as $j\neq 0
$, we have also $s_{j}^{\vee }\leq s_{j}^{\prime }$.

As a result $s_{j}^{\vee }\leq s_{j}^{\prime }$ for all $j\in \mathcal{J}_{0}
$. But as these two vectors both sum to one, we in fact get equality.

}
\end{proof}

{\small Corollary~\ref%
{cor:latticeStructureNoTies} now follows from theorem~\ref{thm:InverseIsotoneNoTies}. Indeed, if $\delta \in \tilde{\sigma}%
^{-1}\left( s\right) $ and $\delta ^{\prime }\in \tilde{\sigma}^{-1}\left(
s\right) $, then }$\delta \wedge \delta ^{\prime }\in \tilde{\sigma}%
^{-1}\left( s\right) $ and $\delta \vee \delta ^{\prime }\in \tilde{\sigma}%
^{-1}\left( s\right) $\hfill QED.

\begin{remark}
{\small To the best of our knowledge, this result is novel in the theory of
two-sided matchings with imperfectly transferable utility. While in the case
of matching with (perfectly) transferable utility, it follows easily from
the fact that the value of the optimal assignment problem is a supermodular
function in $\left( P,-s\right) $, (see e.g. \citeasnoun{vohra2004advanced},
theorem~7.20),  is appears to be novel beyond that case%
\footnote{\citeasnoun{demange1985strategy} show (in their lemma 2) isotonicity in the strong
set order with respect to reservation utilities, which is a different result.%
}. \hfill $\square $ }
\end{remark}
\subsection{{\protect\small Proof of theorem~\protect\ref%
{thm:existenceNoTies}}}

{\small The proof of theorem~\protect\ref%
{thm:existenceNoTies} is based on the additional lemmas \ref{lem:1}-\ref{lem:4}, which we first state and prove. }

{\small Assumption~\ref{ass:leftRightBehav} implies that $\forall\ \eta >0,\
\nu >0$, there is $\delta ^{\ast }$ s.t. $\delta >\delta ^{\ast }$
implies $\Pr \left( \left\vert X_{\delta }-X_{\delta }^{\ast }\right\vert
>\nu \right) <\eta $. }

\begin{lemma}
{\small \label{lem:1}There is a $T^{\ast }$ such that for $T<T^{\ast }$
there exists $\underline{\delta }_{j}^{T}$ such that
\begin{equation}
\int \frac{\exp \left( \frac{\mathcal{U}_{\varepsilon j}\left( \underline{%
\delta }_{j}\right) }{T}\right) }{1+\exp \left( \frac{\mathcal{U}%
_{\varepsilon j}\left( \underline{\delta }_{j}\right) }{T}\right) }P\left(
d\varepsilon \right) =s_{j}  \label{solveForDelta}
\end{equation}%
and for all $T<T^{\ast }$, $\underline{\delta }_{j}^{T}\geq \underline{%
\delta }_{j}$ where $\underline{\delta }_{j}$\ does not depend on $T$. }
\end{lemma}

\begin{proof}
{\small For $T>0$, let
\begin{equation*}
F_{j}^{T}\left( \delta \right) =\int \frac{\exp \left( \frac{\mathcal{U}%
_{\varepsilon j}\left( \delta \right) }{T}\right) }{1+\exp \left( \frac{%
\mathcal{U}_{\varepsilon j}\left( \delta \right) }{T}\right) }P\left(
d\varepsilon \right) =\int \frac{1}{1+\exp \left( -\frac{\mathcal{U}%
_{\varepsilon j}\left( \delta \right) }{T}\right) }P\left( d\varepsilon
\right) .
\end{equation*}
}

{\small Assumption ~\ref{ass:ContinuousMeasurable}, part (b) implies: }

{\small \underline{Fact (a):} $F_{j}^{T}\left( .\right) $ is continuous and
strictly increasing. }

{\small Next, by Assumption ~\ref{ass:leftRightBehav}, there exists $%
\underline{\delta }_{j}$ such that $\delta <\underline{\delta }_{j}$ implies
$\Pr \left( \mathcal{U}_{\varepsilon j}\left( \delta \right) >-a\right) \leq
s_{j}/2$. \ Hence, for $\delta <\underline{\delta }_{j}$
\begin{eqnarray*}
F_{j}^{T}\left( \delta \right) &=&\int_{\left\{ \mathcal{U}_{\varepsilon
j}\left( \delta \right) <-a\right\} }\frac{\exp \left( \frac{\mathcal{U}%
_{\varepsilon j}\left( \delta \right) }{T}\right) }{1+\exp \left( \frac{%
\mathcal{U}_{\varepsilon j}\left( \delta \right) }{T}\right) }P\left(
d\varepsilon \right) +\int_{\left\{ \mathcal{U}_{\varepsilon j}\left( \delta
\right) \geq -a\right\} }\frac{\exp \left( \frac{\mathcal{U}_{\varepsilon
j}\left( \delta \right) }{T}\right) }{1+\exp \left( \frac{\mathcal{U}%
_{\varepsilon j}\left( \delta \right) }{T}\right) }P\left( d\varepsilon
\right) \\
&\leq &\frac{1}{1+\exp \left( \frac{a}{T}\right) }+s_{j}/2
\end{eqnarray*}
}

{\small and taking $T^{\ast }=a/\log \left( 1/s_{j}-1\right) $ if $\log
\left( 1/s_{j}-1\right) >0$, and $T^{\ast }=+\infty $ else, it follows that $%
T\leq T^{\ast }$ implies $\frac{1}{1+\exp \left( \frac{a}{T}\right) }\leq
s_{j}/2$, hence we get to: }

{\small \underline{Fact (b):} for $\delta <\underline{\delta }_{j}$ and $%
T\leq T^{\ast }$, one has $F_{j}^{T}\left( \delta \right) <s_{j}$. }

{\small Next, by Assumption~\ref{ass:leftRightBehav}, there exists $\delta
_{j}^{\prime } $ such that $\delta >\delta _{j}^{\prime }$ implies $\Pr
\left( \mathcal{U}_{\varepsilon j}\left( \delta \right) >0\right) \geq
2s_{j} $. Then for $\delta >\delta _{j}^{\prime }$,
\begin{equation*}
F_{j}^{T}\left( \delta \right) \geq \int_{\left\{ \mathcal{U}_{\varepsilon
j}\left( \delta \right) >b\right\} }\frac{\exp \left( \frac{\mathcal{U}%
_{\varepsilon j}\left( \delta \right) }{T}\right) }{1+\exp \left( \frac{%
\mathcal{U}_{\varepsilon j}\left( \delta \right) }{T}\right) }P\left(
d\varepsilon \right) \geq \frac{\Pr \left( \mathcal{U}_{\varepsilon j}\left(
\delta \right) >b\right) }{1+\exp \left( 0\right) }\geq \frac{2s_{j}}{2}%
=s_{j}.
\end{equation*}
}

{\small As a result, we get: }

{\small \underline{Fact (c):} for all $T>0$ and for $\delta >\delta
_{j}^{\prime }$, $F_{j}^{T}\left( \delta \right) >s_{j}$. }

{\small By combination of facts (a), (b) and (c), it follows that for $T\leq
T^{\ast }$, there exists a unique $\underline{\delta }_{j}^{T}$ such that $%
F_{j}^{T}\left( \underline{\delta }_{j}^{T}\right) =s_{j}$ and $\underline{%
\delta }_{j}^{T}\leq \delta _{j}^{\prime }$, where $\delta _{j}^{\prime }$
does not depend on $T\leq T^{\ast }$. }
\end{proof}

{\small Let
\begin{equation*}
G_{j}^{T}\left( \delta _{j};\delta _{-j}\right) :=\int \frac{P\left(
d\varepsilon \right) }{\exp \left( -\frac{\mathcal{U}_{\varepsilon j}\left(
\delta _{j}\right) }{T}\right) +\sum_{j^{\prime }\in \mathcal{J}}\exp \left(
\frac{\mathcal{U}_{\varepsilon j^{\prime }}\left( \delta _{j^{\prime
}}\right) -\mathcal{U}_{\varepsilon j}\left( \delta _{j}\right) }{T}\right) }%
.
\end{equation*}
}

\begin{lemma}
{\small \label{lem:2}For $T<T^{\ast }$, if $G_{j}^{T}\left( \delta
_{j}^{T,k};\delta _{-j}^{T,k}\right) \leq s_{j}$, then: }

{\small (i) there is a real $\delta _{j}^{T,k+1}\geq \delta _{j}^{T,k}$ such
that%
\begin{equation*}
G_{j}^{T}\left( \delta _{j}^{T,k+1};\delta _{-j}^{T,k}\right) =s_{j},
\end{equation*}
}

{\small (ii) one has $G_{j}^{T}\left( \delta _{j}^{T,k+1};\delta
_{-j}^{T,k+1}\right) \leq s_{j}$. }
\end{lemma}

\begin{proof}
{\small Take $\eta >0$ such that $\eta <1-\sqrt{s_{j}}$. There is $M>0$ such
that
\begin{equation*}
\Pr \left( 1+\sum_{j^{\prime }\neq j}\exp \left( \frac{\mathcal{U}%
_{\varepsilon j^{\prime }}\left( \delta _{j^{\prime }}^{T,k}\right) }{T}%
\right) <M\right) >1-\eta /2.
\end{equation*}
We have%
\begin{eqnarray*}
G_{j}^{T}\left( \delta _{j};\delta _{-j}^{T,k}\right) &\geq &\int \frac{%
1\left\{ 1+\sum_{j^{\prime }\neq j}\exp \left( \frac{\mathcal{U}%
_{\varepsilon j^{\prime }}\left( \delta _{j^{\prime }}^{T,k}\right) }{T}%
\right) <M\right\} P\left( d\varepsilon \right) }{1+\exp \left( -\frac{%
\mathcal{U}_{\varepsilon j}\left( \delta _{j}\right) }{T}\right) \left(
1+\sum_{j^{\prime }\neq j}\exp \left( \frac{\mathcal{U}_{\varepsilon
j^{\prime }}\left( \delta _{j^{\prime }}^{T,k}\right) }{T}\right) \right) }
\\
&\geq &\int \frac{1\left\{ 1+\sum_{j^{\prime }\neq j}\exp \left( \frac{%
\mathcal{U}_{\varepsilon j^{\prime }}\left( \delta _{j^{\prime
}}^{T,k}\right) }{T}\right) <M\right\} P\left( d\varepsilon \right) }{1+\exp
\left( -\frac{\mathcal{U}_{\varepsilon j}\left( \delta _{j}\right) }{T}%
\right) M}
\end{eqnarray*}
}

{\small Next, by Assumption~\ref{ass:leftRightBehav}, for all $b\in \mathbb{R%
}$, there exists $\delta _{j}^{\ast }$ such that $\delta >\delta _{j}^{\ast
} $ implies $\Pr \left( \mathcal{U}_{\varepsilon j}\left( \delta \right)
>b\right) \geq 1-\eta /2$. Thus for $\delta _{j}>\delta _{j}^{\ast }$,
\begin{equation*}
G_{j}^{T}\left( \delta _{j};\delta _{-j}^{T,k}\right) \geq \int \frac{%
1\left\{ 1+\sum_{j^{\prime }\neq j}\exp \left( \frac{\mathcal{U}%
_{\varepsilon j^{\prime }}\left( \delta _{j^{\prime }}^{T,k}\right) }{T}%
\right) <M\right\} 1\left\{ \mathcal{U}_{\varepsilon j}\left( \delta
_{j}\right) >b\right\} P\left( d\varepsilon \right) }{1+\exp \left( -\frac{b%
}{T}\right) M}\geq \frac{1-\eta }{1+\exp \left( -\frac{b}{T}\right) M}.
\end{equation*}%
Choosing $b=-T\log \left( \eta /M\right) $ implies that the right hand-side
is $\frac{1-\eta }{1+\eta }\geq \left( 1-\eta \right) ^{2}$. Because $\eta
<1-\sqrt{s_{j}}$, $\left( 1-\eta \right) ^{2}>s_{j}$, and therefore for $%
\delta _{j}>\delta _{j}^{\ast }$, $G_{j}^{T}\left( \delta _{j};\delta
_{-j}^{T,k}\right) >s_{j}$. Hence, because $G_{j}^{T}\left( \delta
_{j}^{T,k};\delta _{-j}^{T,k}\right) \leq s_{j}$, by continuity of $%
G_{j}^{T}\left( .;\delta _{-j}^{T,k}\right) $, there exists $\delta
_{j}^{T,k+1}\in (\delta _{j}^{T,k},\delta _{j}^{\ast })$ such that
\begin{equation*}
G_{j}^{T}\left( \delta _{j}^{T,k+1};\delta _{-j}^{T,k}\right) =s_{j},
\end{equation*}%
which shows claim (i). To show the second claim, let us note that $%
G_{j}^{T}\left( \delta \right) $ is decreasing with respect to $\delta
_{j^{\prime }}$ for any $j^{\prime }\neq j$. Indeed,
\begin{equation*}
G_{j}^{T}\left( \delta _{j};\delta _{-j}\right) =\int \frac{P\left(
d\varepsilon \right) }{\exp \left( -\frac{\mathcal{U}_{\varepsilon j}\left(
\delta _{j}\right) }{T}\right) +1+\sum_{j^{\prime }\neq j}\exp \left( \frac{%
\mathcal{U}_{\varepsilon j^{\prime }}\left( \delta _{j^{\prime }}\right) -%
\mathcal{U}_{\varepsilon j}\left( \delta _{j}\right) }{T}\right) }
\end{equation*}%
is expressed as the expectation of a term which is decreasing in $\delta
_{j^{\prime }}$. Hence, as $\delta _{-j}^{T,k}\leq \delta _{-j}^{T,k+1}$ in
the componentwise order, it follows that%
\begin{equation*}
G_{j}^{T}\left( \delta _{j}^{T,k+1};\delta _{-j}^{T,k+1}\right) \leq
G_{j}^{T}\left( \delta _{j}^{T,k+1};\delta _{-j}^{T,k}\right) =s_{j},
\end{equation*}%
which shows claim (ii). }
\end{proof}

{\small Because of lemma~\ref{lem:2}, one can construct recursively a
sequence $\left( \delta _{j}^{T,k}\right) $ such that $\delta
_{j}^{T,k+1}\geq \delta _{j}^{T,k}$ and
\begin{equation}
G_{j}^{T}\left( \delta _{j}^{T,k}\right) \leq s_{j},  \label{ineqToSum}
\end{equation}
}

{\small From Assumption~\ref{ass:leftRightBehav}, setting $\eta =s_{0}/4$
and $b=T^{\ast }\log \left( 4/s_{0}-1\right) $, one has the existence of $%
\delta \in \mathbb{R}$ such that $\delta \geq \overline{\delta }_{j}$
implies $\Pr \left( \mathcal{U}_{\varepsilon j}\left( \delta \right)
<b\right) <\eta $. }

\begin{lemma}
{\small \label{lem:3}For all $k\in \mathbb{N}$ and $T<T^{\ast }$, one has
\begin{equation}
\delta _{j}^{k}\leq \overline{\delta }_{j}  \label{IneqQED}
\end{equation}%
where $\overline{\delta }_{j}$ is a constant independent from $T<T^{\ast }$.
}
\end{lemma}

\begin{proof}
{\small By summation of inequality~(\ref{ineqToSum}) over $j\in \mathcal{J}$%
, one has%
\begin{eqnarray*}
s_{0} &\leq &\int \frac{P\left( d\varepsilon \right) }{1+\sum_{j^{\prime
}\in \mathcal{J}}\exp \left( \mathcal{U}_{\varepsilon j^{\prime }}\left(
\delta _{j^{\prime }}^{T,k}\right) /T\right) }\leq \int \frac{P\left(
d\varepsilon \right) }{1+\exp \left( \mathcal{U}_{\varepsilon j}\left(
\delta _{j}^{T,k}\right) /T\right) } \\
&\leq &\Pr \left( \mathcal{U}_{\varepsilon j}\left( \delta _{j}^{T,k}\right)
<b\right) +\int_{\left\{ \mathcal{U}_{\varepsilon j}\left( \delta
_{j}^{T,k}\right) \geq b\right\} }\frac{P\left( d\varepsilon \right) }{%
1+\exp \left( \mathcal{U}_{\varepsilon j}\left( \delta _{j}^{T,k}\right)
/T\right) } \\
&\leq &\Pr \left( \mathcal{U}_{\varepsilon j}\left( \delta _{j}^{T,k}\right)
<b\right) +\frac{1}{1+\exp \left( b/T^{\ast }\right) }.
\end{eqnarray*}%
Now assume by contradiction that $\delta _{j}^{T,k}>\overline{\delta }_{j}$.
Then $\Pr \left( \mathcal{U}_{\varepsilon j}\left( \delta \right) <b\right)
<\eta =s_{0}/4$ and $\left( 1+\exp \left( b/T^{\ast }\right) \right)
^{-1}=s_{0}/4$, and thus one would have
\begin{equation*}
s_{0}\leq s_{0}/4+s_{0}/4=s_{0}/2\text{,}
\end{equation*}%
a contradiction. Thus inequality~(\ref{IneqQED}) holds. }
\end{proof}

\begin{lemma}
{\small \label{lem:4}Let $\delta _{j}^{T}=\lim_{k\rightarrow +\infty }\delta
_{j}^{T,k}$. One has%
\begin{equation}
G_{j}^{T}\left( \delta _{j}^{T};\delta _{-j}^{T}\right) =s_{j}.
\label{limitToShow}
\end{equation}
}
\end{lemma}

\begin{proof}
{\small One has $G_{j}^{T}\left( \delta _{j}^{T,k+1};\delta
_{-j}^{T,k}\right) =s_{j} $; by the fact that $\delta
_{j}^{T,k+1}\rightarrow \delta _{j}^{T}$ and $\delta _{-j}^{T,k}\rightarrow
\delta _{-j}^{T}$ and by the continuity of $G_{j}^{T}$, it follows~(\ref%
{limitToShow}). }
\end{proof}

{\small We can now deduce the proof of theorem~\ref{thm:existenceNoTies}: }

\begin{proof}[Proof of theorem~\ref{thm:existenceNoTies}]
{\small 
Lemma~\ref{lem:4} implies that one can define
\begin{equation*}
u^{T}\left( \varepsilon \right) =T\log \left( 1+\sum_{j\in \mathcal{J}}\exp
\left( \mathcal{U}_{\varepsilon j}\left( \delta _{j}^{T}\right) /T\right)
\right) \text{ and }\pi _{\varepsilon j}^{T}=\exp \left( \frac{-u^{T}\left(
\varepsilon \right) +\mathcal{U}_{\varepsilon j}\left( \delta
_{j}^{T}\right) }{T}\right) ,
\end{equation*}%
and by the same result, one has%
\begin{equation*}
\mathbb{E}_{\pi ^{T}}\left[ u^{T}\left( \varepsilon \right) \right] =\mathbb{%
E}_{\pi ^{T}}\left[ \mathcal{U}_{\varepsilon j}\left( \delta _{j}^{T}\right) %
\right] .
\end{equation*}%
It follows from lemma~\ref{lem:3} that the sequence $\delta _{j}^{T}$ is
bounded independently of $T$, so by compactness, it converge up to
subsequence toward $\delta _{j}^{0}$. Note that \underline{$\delta $}$%
_{j}^{0}\leq \delta _{j}^{0}\leq \overline{\delta }_{j}^{0}$. We can extract
a converging subsequence $\pi ^{T_{n}}$ where $T_{n}\rightarrow 0$ and $\pi
^{T_{n}}\rightarrow \pi ^{0}$ in the weak convergence. Mimicking the
argument in Villani (2003) page 32, it follows that $\pi ^{0}\in \mathcal{M}%
\left( P,s\right) $. }

{\small 
Let $u^{0}\left( \varepsilon \right)
=\max_{j\in \mathcal{J}_{0}}\left\{ \mathcal{U}_{\varepsilon j}\left( \delta
_{j}^{0}\right) \right\} $. We have $u^{0}\left( \varepsilon \right) \geq
\mathcal{U}_{\varepsilon j}\left( \delta _{j}^{0}\right) $. Let us show that
\begin{equation*}
\mathbb{E}_{\pi ^{0}}\left[ u^{0}\left( \varepsilon \right) \right] =\mathbb{%
E}_{\pi ^{0}}\left[ \mathcal{U}_{\varepsilon J}\left( \delta _{J}^{0}\right) %
\right] ,
\end{equation*}%
which will prove the final result. We have $\mathbb{E}_{\pi ^{T_{n}}}\left[
u^{T_{n}}\left( \varepsilon \right) \right] =\mathbb{E}_{\pi ^{T_{n}}}\left[
\mathcal{U}_{\varepsilon j}\left( \delta _{j}^{T_{n}}\right) \right] $; let
us show that%
\begin{equation*}
\begin{tabular}{l}
(i)~$\mathbb{E}_{\pi ^{T_{n}}}\left[ u^{T_{n}}\left( \varepsilon \right) %
\right] \rightarrow \mathbb{E}_{\pi ^{0}}\left[ u^{0}\left( \varepsilon
\right) \right] \text{, and}$ \\
(ii)~$\mathbb{E}_{\pi ^{T_{n}}}\left[ \mathcal{U}_{\varepsilon J}\left(
\delta _{J}^{T_{n}}\right) \right] \rightarrow \mathbb{E}_{\pi ^{0}}\left[
\mathcal{U}_{\varepsilon J}\left( \delta _{J}^{0}\right) \right] $%
\end{tabular}%
\end{equation*}
}

{\small Start by showing point (i). We have $0\leq u^{0}\left( \varepsilon
\right) -u^{T}\left( \varepsilon \right) \leq T_{n}\log \mathcal{J}$. As a
result, $\mathbb{E}_{\pi ^{T_{n}}}\left[ u^{T_{n}}\left( \varepsilon \right) %
\right] =\mathbb{E}_{P}\left[ u^{T_{n}}\left( \varepsilon \right) \right]
\rightarrow \mathbb{E}_{P}\left[ u^{0}\left( \varepsilon \right) \right] =%
\mathbb{E}_{\pi ^{0}}\left[ u^{0}\left( \varepsilon \right) \right] $. }

{\small Next, we show point (ii). One has,%
\begin{equation*}
\mathbb{E}_{\pi ^{T_{n}}}\left[ \mathcal{U}_{\varepsilon J}\left( \delta
_{J}^{T_{n}}\right) \right] -\mathbb{E}_{\pi ^{0}}\left[ \mathcal{U}%
_{\varepsilon J}\left( \delta _{J}^{0}\right) \right] =\mathbb{E}_{\pi
^{T_{n}}}\left[ \mathcal{U}_{\varepsilon J}\left( \delta _{j}^{T_{n}}\right)
-\mathcal{U}_{\varepsilon J}\left( \delta _{J}^{0}\right) \right] +\mathbb{E}%
_{\pi ^{T_{n}}}\left[ \mathcal{U}_{\varepsilon J}\left( \delta
_{J}^{0}\right) \right] -\mathbb{E}_{\pi ^{0}}\left[ \mathcal{U}%
_{\varepsilon J}\left( \delta _{J}^{0}\right) \right]
\end{equation*}%
Let $\nu >0$. For any $K\subseteq \mathcal{X}$ compact subset of $\mathcal{X}
$, one has%
\begin{equation*}
\left\vert \mathbb{E}_{\pi ^{T_{n}}}\left[ \mathcal{U}_{\varepsilon J}\left(
\delta _{J}^{T_{n}}\right) -\mathcal{U}_{\varepsilon J}\left( \delta
_{J}^{0}\right) \right] \right\vert \leq \mathbb{E}_{\pi ^{T_{n}}}\left[
\mathcal{U}_{\varepsilon J}\left( \delta _{J}^{T_{n}}\right) -\mathcal{U}%
_{\varepsilon J}\left( \delta _{J}^{0}\right) 1_{\left\{ \varepsilon \in
K\right\} }\right] +2\mathbb{E}_{P}\left[ \sum_{j}\left\vert \mathcal{U}%
_{\varepsilon j}\left( \overline{\delta }_{j}\right) \right\vert 1\left\{
\varepsilon \in K\right\} \right]
\end{equation*}%
hence, one may choose $K$ such that $\mathbb{E}_{P}\left[ \sum_{j}\left\vert
\mathcal{U}_{\varepsilon j}\left( \overline{\delta }_{j}\right) \right\vert
1\left\{ \varepsilon \in K\right\} \right] <\nu /4$. By uniform continuity
of $\varepsilon \rightarrow \mathcal{U}_{\varepsilon j}\left( \delta \right)
$ on $K$, and because $\delta _{j}^{T_{n}}\rightarrow \delta _{j}^{0}$,
there exists $n^{\prime }\in \mathbb{N}$ such that $n\geq \overline{n}$
implies $\max_{j\in \mathcal{J}}\left\vert \mathcal{U}_{\varepsilon J}\left(
\delta _{j}^{T_{n}}\right) -\mathcal{U}_{\varepsilon J}\left( \delta
_{j}^{0}\right) \right\vert \leq \nu /2$ for each $\varepsilon \in K$. Thus,
for $n\geq n^{\prime }$, one has%
\begin{equation}
\left\vert \mathbb{E}_{\pi ^{T_{n}}}\left[ \mathcal{U}_{\varepsilon J}\left(
\delta _{J}^{T_{n}}\right) -\mathcal{U}_{\varepsilon J}\left( \delta
_{J}^{0}\right) \right] \right\vert \leq \nu  \label{maj1}
\end{equation}%
By the weak convergence of $\pi ^{T_{n}}$ toward $\pi ^{0}$, there is $%
n^{\prime \prime }\geq n^{\prime }$ such that for $n\geq n^{\prime \prime }$
one has%
\begin{equation}
\left\vert \mathbb{E}_{\pi ^{T_{n}}}\left[ \mathcal{U}_{\varepsilon J}\left(
\delta _{J}^{0}\right) \right] -\mathbb{E}_{\pi ^{0}}\left[ \mathcal{U}%
_{\varepsilon J}\left( \delta _{J}^{0}\right) \right] \right\vert \leq \nu .
\label{maj2}
\end{equation}%
Combining~(\ref{maj1}) and~(\ref{maj2}), it follows that for $n\geq
n^{\prime \prime }$,%
\begin{equation*}
\left\vert \mathbb{E}_{\pi ^{T_{n}}}\left[ \mathcal{U}_{\varepsilon J}\left(
\delta _{J}^{T_{n}}\right) -\mathcal{U}_{\varepsilon J}\left( \delta
_{J}^{0}\right) \right] \right\vert \leq 2\nu ,
\end{equation*}%
which establishes point (ii). The result is proven by noting that $\mathbb{E}%
_{\pi ^{0}}\left[ u^{0}\left( \varepsilon \right) \right] =\mathbb{E}_{\pi
^{0}}\left[ \mathcal{U}_{\varepsilon J}\left( \delta _{J}^{0}\right) \right]
$ along with $u^{0}\left( \varepsilon \right) \geq \mathcal{U}_{\varepsilon
j}\left( \delta _{j}^{0}\right) $ for all $\varepsilon $ and $j$ implies
that $\left( \varepsilon ,j\right) \in Supp\left( \pi ^{0}\right) $ implies $%
u^{0}\left( \varepsilon \right) =\mathcal{U}_{\varepsilon j}\left( \delta
_{j}^{0}\right) $. }
\end{proof}

\subsection{{\protect\small Proof of theorem~\protect\ref%
{thm:uniquenessNoTies}}}

\begin{proof}
{\small We have
\begin{equation*}
s_{0}=P\left( Z_{0j}\geq \delta _{j} \forall j \in \mathcal{J}\right) .
\end{equation*}%

 If $Z$ has full density, it follows that $s_0$ is strictly decreasing in any $\delta_j$, and hence, $0$ is a strict substitute to any other $j \in \mathcal{J}$ in the sense of~ \citeasnoun{berry2013connected}, and 
therefore the conclusion holds by the result in that paper.}
\end{proof}

\subsection{{\protect\small Proof of theorem~\protect\ref{thm:consistency}}}

\begin{proof}
Define $Z_{jk}=U_{\varepsilon j}^{-1}\left( U_{\varepsilon k}\left( \delta
_{k}\right) \right) $, which extends the definition of $Z_{0}$ to any other $%
k$, and we denote $F_{jk}^{n}$ its c.d.f. under $P^{n}$. Recall $%
Z_{jk}=U_{\varepsilon j}^{-1}\left( U_{\varepsilon k}\left( \delta
_{k}\right) \right) $ and denote $F_{j0}^{n}$ its c.d.f. under $P^{n}$. If $%
\delta ^{n}\in \tilde{\sigma}^{-1}\left( P_{n},\delta ^{n}\right)$, then
\begin{equation*}
\min_{n}\left\{ s_{0}^{n}\right\} \leq s_{0}^{n}\leq P^{n}\left(
U_{\varepsilon j}^{-1}\left( U_{\varepsilon 0}\left( \delta _{0}\right)
\right) \geq \delta _{j}^{n}\right) =1-F_{j0}^{n}\left( \delta
_{j}^{n}\right) 
\end{equation*}
thus $\delta _{j}^{n}\leq \left( F_{j0}^{n}\right)
^{-1}\left( 1-\min_{n^{\prime }}\left\{ s_{0}^{n^{\prime }}\right\} \right)
\leq \max_{n}\left\{ \left( F_{j0}^{n}\right) ^{-1}\left( 1-\min_{n^{\prime
}}\left\{ s_{0}^{n^{\prime }}\right\} \right) \right\} $ and similarly
\begin{equation*}\min_{n}\left\{ s_{j}^{n}\right\} \leq s_{j}^{n}\leq P^{n}\left( \delta
_{j}^{n}\geq U_{\varepsilon j}^{-1}\left( U_{\varepsilon 0}\left( \delta
_{0}\right) \right) \right) =F_{j0}^{n}\left( \delta _{j}^{n}\right) 
\end{equation*}
thus $%
\delta _{j}^{n}\geq \left( F_{j0}^{n}\right) ^{-1}\left( 1-\min_{n^{\prime
}}\left\{ s_{0}^{n^{\prime }}\right\} \right) \geq \min_{n}\left\{ \left(
F_{j0}^{n}\right) ^{-1}\left( 1-\min_{n^{\prime }}\left\{ s_{0}^{n^{\prime
}}\right\} \right) \right\} $. Therefore $\delta ^{n}$ remains bounded.

Now consider $\delta ^{\ast }$ the limit of a converging subsequence $\delta
^{n}$. The (sub)sequence $\left( \delta _{n},\varepsilon _{n}\right) $
converges weakly towards $\left( \delta ^{\ast },\varepsilon \right) $, and
by the virtues of Lemma 2.2, part (vii) in \citeasnoun{vandervaart} given
assumption~\ref{ass:NoIndiff}, get that
\[
s_{j}=\lim_{n}P\left( U_{\varepsilon ^{n}j}\left( \delta _{j}^{n}\right)
\geq U_{\varepsilon ^{n}k}\left( \delta _{k}^{n}\right) \forall k\in
\mathcal{J}_{0}\right) =P\left( U_{\varepsilon j}\left( \delta _{j}^{\ast
}\right) \geq U_{\varepsilon k}\left( \delta _{k}^{\ast }\right) \forall
k\in \mathcal{J}_{0}\right)
\]%
and therefore, by the fact that $\delta ^{\ast }$ is the unique solution to
the demand inversion problem as proven in theorem~\ref{thm:uniquenessNoTies}
above, we get that $\delta ^{\ast }=\delta $. As $\delta ^{n}$ is bounded
and has a unique adherence value, it converges to $\delta $.

\end{proof}
{\small 
}

{\small 
}

\subsection{\protect\small Market Shares Adjustment: Algorithm for the lower bound}\label{algo for the lower bound}

In order to calculate the lower bound, one could implement the same algorithm as the one for the upper bound, but invert the roles of yogurts and consumers as in \citeasnoun{kelso1982job}. However, this would be inefficient as there are few brands of yogurts but many different consumers. Therefore, the algorithm would be fast for the upper bound, as it
deals with only few $\delta_{j}$'s, but not for the lower bound as it deals with a lot of different $u_{\varepsilon}$. Instead of switching the roles of
consumers and yogurts, we adapt the upper bound algorithm described in section \ref{sec:algorithms} for the lower bound.

We set the initial systematic utility $\delta_j^{ub}$ equal to the lattice upper bound (estimated using the algorithm for the upper bound).  In the ``first loop'' below, we iterate from $\left\{ \delta_j^{ub}\right\}$ down to values of $\delta$ which are below the lower bound $\underline{\delta}$ (corresponding to a vector of $\delta$ at which all the non-reference brands $j\neq 0$ are in excess supply and the reference brand $0$ is in excess demand).   In the ``second loop'', we iterate up from this point up to the lower bound, similarly to the MSA upper bound algorithm.

\begin{algorithm}[MSA lower bound]
\label{algo:msa lower bound}

Take $\eta ^{init}=1$, $\delta _{j}^{init}=\delta_j^{ub}$ and $block_{j}=0$ for all $j \in \mathcal{J}$.

\underline{Begin first loop}

\qquad Require $\left( \delta _{j}^{init},\eta ^{init},block_{j}\right) .$

\qquad Set $\eta =\eta ^{init}$ and $\delta ^{0}=\delta ^{init}$.

While  $\eta \geq \eta ^{tol}$

\qquad Set $\pi _{ij}=1$ if $j\in \arg \max_{j}\mathcal{U}%
_{\varepsilon j}(\delta _{j})$, and $=0$ otherwise (breaking ties
arbitrarily).

\qquad If $\sum_j block_j = |\mathcal{J} |$, then set $\delta _{j}=
\delta _{j}+2\eta $ and $block_j = 0$ for all $j\in \mathcal{J}$, and $\eta = \eta /4.
$

\qquad Else set

\qquad \qquad  $\delta _{j}= \delta _{j}-\eta \mathbbm{1}\left\{
\sum\limits_{i}\pi _{ij} \geq m_{j}\right\} $ for all $j\in \mathcal{J}$  

\qquad \qquad If $\mathbbm{1}\left\{
\sum\limits_{i}\pi _{ij} < m_{j}\right\}$, then $block_j =  \mathbbm{1}\left\{
\sum\limits_{i}\pi _{i0} > m_{0}\right\}$

\qquad \qquad Else $block_j = block_j$

End While

Return $\delta ^{return}=\delta $.

\underline{End first loop}

\vspace{0.5cm}

\underline{Begin main second loop}

Take $\eta=\eta^{tol}$ and $\delta _{j}^{init}=\delta_j^{return}$.

Repeat:

\qquad Call the {\em inner second loop} with parameter values $\left( \delta
_{j}^{init}\right) $  which returns $\left( \delta
^{return}\right) $.

\qquad Set $\delta ^{init}= \delta ^{return}- 2\eta ^{tol}$

Until $\delta _{j}^{return}>\delta _{j}^{init}$ for all $j \in \mathcal{J}$.

\underline{End main second loop}

\bigskip

\underline{Begin inner second loop}

\qquad Require $\left( \delta _{j}^{init}\right) .$

\qquad Set $\delta =\delta ^{init}$.

\qquad Set $\pi _{i0}=1$ if $0 \in \arg \max_{j \in \mathcal{J}_0} \mathcal{U}%
_{\varepsilon j}(\delta _{j})$, and $=0$ otherwise (breaking ties
arbitrarily).

While  $\sum\limits_{i}\pi _{i0} > m_{0}$

\qquad Set $\pi _{ij}=1$ if $j\in \arg \max_{j}\mathcal{U}%
_{\varepsilon j}(\delta _{j})$, and $=0$ oherwise (breaking ties
arbitrarily).

\qquad Set $\delta _{j}= \delta _{j} + \eta \mathbbm{1}\left\{
\sum\limits_{i}\pi _{ij}<m_{j}\right\} $ for all $j\in \mathcal{J}$.

\qquad Set $\pi _{i0}=1$ if $0 \in \arg \max_{j \in \mathcal{J}_0} \mathcal{U}%
_{\varepsilon j}(\delta _{j})$, and $=0$ otherwise (breaking ties
arbitrarily).

End While

Return $\delta ^{return}=\delta $.

\underline{End inner second loop}

\end{algorithm}

\section{Semi-discrete transport algorithms}\label{sec:semidiscrete}
In this section we describe an additional computational algorithm which is specialized for solving the particular case of the pure characteristics demand model (discussed as Example 2.1 above). The method discussed here can be used when the distribution of
the unobserved taste vector is uniformly distributed over a polyhedron,
typically the Cartesian product of compact intervals. Recall that the pure
characteristics model has $\varepsilon _{ij}=\nu _{i}^{\intercal }x_{j}$,
with $\nu \sim \mathbf{P}_{\nu }$ is a random vector distributed over $%
\mathbb{R}^{d}$, and assume that $\mathbf{P}_{\nu }$ is the uniform
distribution over $E=\prod_{1\leq k\leq d}\left[ 0,l_{k}\right] $. Then we
can use the equivalence theorem in order to compute $\tilde{\sigma}^{-1}$
using semi-discrete transport algorithms, which were pioneered by~%
\citeasnoun{aurenhammer1987power}, with substantial progress made recently by~\citeasnoun{kitagawa2016convergence} and~%
\citeasnoun{levy2015numerical}. The idea, exposited in chapter 5 of~%
\citeasnoun{galichon2015optimal}, is that the optimal transport problem~(\ref%
{semi-discreteTransport}) can be reformulated as a finite-dimensional
unconstrained convex optimization problem%
\begin{equation}
\inf\limits_{\delta \in \mathbb{R}^{\mathcal{J}_{0}}}F\left( \delta \right)
\text{, where }F\left( \delta \right) =\mathbb{E}_{P}\left[ \max_{j\in
\mathcal{J}_{0}}\left\{ \delta _{j}+\varepsilon _{j}\right\} \right]
-\sum_{j\in \mathcal{J}_{0}}\delta _{j}s_{j}.
\label{finiteDimensionalOptTransport}
\end{equation}

Semi-discrete algorithms consist of a gradient descent over $%
F $. Note that $\partial F/\partial \delta _{j}=\Pr \left( \forall j^{\prime
}\in \mathcal{J}_{0}\backslash \left\{ j\right\} ,\varepsilon _{j^{\prime
}}-\varepsilon _{j}\leq \delta _{j}-\delta _{j^{\prime }}\right) -s_{j}$,
where the first term is the area of the polytope $\left\{ \varepsilon \in
E:\forall j^{\prime }\in \mathcal{J}_{0}\backslash \left\{ j\right\}
,\varepsilon _{j^{\prime }}-\varepsilon _{j}\leq \delta _{j}-\delta
_{j^{\prime }}\right\} $, hence a gradient descent can be done provided one
can compute areas of polytopes. The Hessian of $F$ can be computed
relatively easily too; we refer to~\citeasnoun{kitagawa2016convergence} for
details.

For these semi-discrete approaches, we provide an R-based interface to Geogram by \citeasnoun{levy2018numerical}. Geogram is a \textsc{C++} library of geometric algorithms, with fast implementations of semi-discrete optimal transport methods with two or three random coefficients.  The package is open-source and is available on GitHub as \citeasnoun{github_rgeogram}.

For the Pure Characteristics Model, the semi-discrete algorithm provides super-fast performance, far outstripping all the other algorithms (cf. table \ref{table:pure char uniform noise} below). Moreover, the semi-discrete algorithm does not rely on simulations, and computes the exact solution set of $\delta_j$ which ensures equilibrium in the market for given market shares and given vector of observed characteristics $x_j$.  The drawbacks are that, as far as we are aware, the semi-discrete approach is only available for random coefficient distributions $\nu_i$ which are jointly-uniform with at most three dimensions; which limits its use in many applications of the pure characteristics models (which typically assumed a joint Gaussian distribution for the random coefficients). In table \ref{table:pure char uniform noise}, we compute Monte-Carlo simulations similar to the ones showed in table \ref{table:pure_char_alg_inner} except that the random tastes shocks $\nu_i$ are drawn from independent uniform distributions and not from Gaussian ones.

\begin{table}[!htbp]
\centering
\begin{threeparttable}
\centering
\begin{tabular}{l|cccc}
 Algorithms & Draws & Brands & RMSE & CPU time (secs) \\ \toprule
 BLP contract. map. & 1,000 & 5 & 0.062 & 0.034  \\
  LP (gurobi) & 1,000 & 5 & 0.009 & 0.104  \\
  Auction & 1,000 & 5 & 0.010 & 0.009  \\
  Semi discrete & 1,000 & 5 & 0 & 0.003 \\
  MSA & 1,000 & 5 & 0.009 & 5.422  \\ \hline
  BLP contract. map. & 1,000 & 50 & 0.032 & 0.283  \\
  LP (gurobi) & 1,000 & 50 & 0.004 & 0.358\\
Auction & 1,000 & 50 & 0.004 & 0.044  \\
  Semi discrete & 1,000 & 50 & 0 & 0.046  \\ \hline
  BLP contract. map. & 1,000 & 500 & 0.011 & 2.766  \\
  LP (gurobi) & 1,000 & 500 & 0.001 & 3.392  \\
  Auction & 1,000 & 500 & 0.001 & 0.480  \\
  Semi discrete & 1,000 & 500 & 0 & 0.743  \\  \hline
  BLP contract. map. & 10,000 & 5 & 0.061 & 0.331  \\
LP (gurobi) & 10,000 & 5 & 0.003 & 0.311 \\
  Auction & 10,000 & 5 & 0.003 & 0.122  \\
  Semi discrete & 10,000 & 5 & 0 & 0.003  \\
  MSA & 10,000 & 5 & 0.003 & 1.471  \\ \hline
  BLP contract. map. & 10,000 & 50 & 0.032 & 2.894  \\
  LP (gurobi) & 10,000 & 50 & 0.001 & 3.947  \\
Auction & 10,000 & 50 & 0.001 & 0.735 \\
  Semi discrete & 10,000 & 50 & 0 & 0.046  \\ \hline
  BLP contract. map. & 10,000 & 500 & 0.011 & 33.171  \\
  LP (gurobi) & 10,000 & 500 & 0.000 & 54.824  \\
  Auction & 10,000 & 500 & 0.000 & 5.451  \\
  Semi discrete & 10,000 & 500 & 0 & 0.749 \\
   \hline
\end{tabular}
\begin{tablenotes}
\footnotesize{
\item Note: The numbers are average of $50$ Monte-Carlo replication. Demand inversion for the pure characteristics model with $5$, $50$ and $500$ brands of yogurt and $1,000$ and $10,000$ draws of taste shocks.  The column "RMSE" corresponds to the root mean squared of in the estimation of the estimated $\delta_j$. Semi-discrete displays $0$ RMSE since there are no sampling errors.}
\end{tablenotes}
\end{threeparttable}
\caption{Average computational Time (secs.) for pure characteristics models with uniform random shocks}
\label{table:pure char uniform noise}
\end{table}

\newpage
\section{Additional Details for Numerical Exercises in Section \ref{sec: numerical pure char}}\label{sec:DGP_pure_char}

In this section, we provide additional details for the results in Section \ref{sec: numerical pure char}. 

Our DGP for Table \ref{table:pure_char_alg_overall} is adapted from \citeasnoun{dube2012improving}. We first generate the market- and product-specific regressors $x_{mj}=(x_{mj1},x_{mj2},x_{mj3})$ from multivariate normal with

\begin{align}
\mu=\begin{pmatrix}
0.5\\
0.5\\
0.5
\end{pmatrix}, \quad \Sigma=\begin{pmatrix}
1&-0.7&0.3\\
\cdot&1&0.3\\
\cdot&\cdot&1
\end{pmatrix}.
\end{align}

\medskip
\noindent The unobserved fixed effect $\xi_{mj}$ is independently generated from a normal distribution with mean equals zero and standard deviation equals 1. The price, $p_{mj}$, is
generated according to
\begin{align}
\begin{array}{l}
p_{mj}=|1.1(x_{mj1}+x_{mj2}+x_{mj3})+0.5\xi_{mj}+e_{mj}|.
\end{array}
\end{align}

\noindent where $e_{mj}$ is independently generated from a normal distribution with mean zero and standard deviation 1. The utility of consumer $i$ who chooses alternative $j$ in market $m$ is generated by

\begin{equation}
    u_{mij}=\beta_0-\beta_p p_{mj}+\beta_1 x_{mj1}+\beta_2 x_{mj2}+\beta_3x_{mj3}+\xi_{mj}.
\end{equation}
\noindent We set $\beta_0=1$. $(\beta_p,\beta_1,\beta_2,\beta_3)$ are individual-specific coefficients generated from independent normal distributions with the means equal (-1,0.5,0.5,0.2) and the standard deviations all equal to 1. As the purpose is comparing the numerical performance, following \citeasnoun{dube2012improving}, the same set of simulated consumers are applied to both the data simulation and all estimation algorithms.

Next, we describe the DGP for instrumental variables. We first generate six basis instrument variables $z$, independently from the following specification:

\begin{align}
\begin{array}{l}
z_{mj}=0.25\left(1.1\left(x_{mj1}+x_{mj2}+x_{mj3}\right)+e_{mj}\right)+u_{mj},
\end{array}
\end{align}

\noindent where $u$ follows a uniform distribution on the unit interval. We use the linear term of $x$ and $z$, the quadratic and cubic terms of $x$ and $z$, the product terms ($\Pi_{k=1}^3x_{mjk}$ and $\Pi_{k=1}^6z_{mjk}$), and the interaction terms, ($x_{mjl}z_{mjk}$, $l=\{1,2\}, k=\{1,\dots,6\}$). There are total 41 instrumental variables.

Below we provide the setup of tuning parameters for various algorithms involved in Table \ref{table:pure_char_alg_overall}. For the inner loop of Matching-LP, we use the combined LP formulation and solved it by GUROBI 9.0. For MPCC and BLP-MPEC, we use KNITRO 12.0. Our choice of the solver depends on the natural of the problem to achieve the best performance. GUROBI is optimized for LP, whereas KNITRO is a general-purpose nonlinear program solver with capability of handling complementarity constraints in \citeasnoun{pang2015constructive}. KNITRO is also recommended by \citeasnoun{dube2012improving} for estimating the mixed logit demand. We program all these three algorithms in AMPL and execute them from AMPL's R interface. The Matching-LP shares similar computational features as in BLP. As illustrated by \citeasnoun{nevo2000practitioner}, one essentially minimizes the scale parameters only: Given the  the scale parameters, one applies the demand inversion to obtain $\delta$ as the ``dependent variable'' and perform the constrained 2SLS (due to the normalization) to obtain the location parameter estimate. For Model I in Table \ref{table:pure_char_alg_overall}, the problem boils down to a convex programming problem. For Model II in Table \ref{table:pure_char_alg_overall}, we use the {\bf optimize} function in R to find the optimal $\sigma_p$. It is based on the golden section search. We set [0.001,5] as the search interval. For MPCC, we deploy the Intel Pardiso MKL in KNITRO using 16 threads. This is to ensure an equal footing since GUROBI used in Matching-LP automatically deploys a parallel solver. Since it is extremely costly to run MPCC, we set a low tolerance in KNITRO (xtol and ftol to 1e-04) and use only one starting point. For BLP-MPEC, we use the logit-smoothed AR simulator.\footnote{See \citeasnoun{Train2009}. \citeasnoun{berry2007pure} also use the same method.} to approximate the demand map. One first generates the individual-level simulators for each product characteristics: $(v_{mip},v_{mi1},v_{mi2},v_{mi3})$. Then we solve the following MPEC formulation of the pure characteristics model:\footnote{We use 10 starting points, automatically chosen by KNITRO. The smoothing parameter of BLP-MPEC is $\lambda=1$.} 

\begin{align}
\begin{array}{rl}\label{GMM_BLP_MPEC}
\underset{\delta,\beta,\sigma}{\min} & g^{'}Wg\\
\textrm{s.t.}& g_n=\sum_{m=1}^M\sum_{j=1}^{J}(\delta_{mj}-x_{mj}^{'}\beta+\alpha p_{mj})z_{mjn}\\[10pt]
&\log(s_{mj})=\log\left(\frac{1}{I}\sum_{i=1}^{I}\left(\dfrac{\exp\left([\delta_{mj}+\sigma_{p}p_{mk}v_{mip}+\sum_{k=1}^{K}\sigma_k x_{mj}v_{mik}]/\lambda\right)}{\sum_{j^{'}\in\mathcal{J}}\exp\left([\delta_{mj^{'}}+\sigma_{p}p_{mk}v_{mip}+\sum_{k=1}^{K}\sigma_k x_{mj}v_{mik}]/\lambda\right)}\right)\right)
\end{array}
\end{align}


The DGP for Table \ref{table:pure_char_alg_inner} is described below: The $x_{j}=(x_{j1},x_{j2},x_{j3})$ are drawn from multivariate normal with
\begin{align}
\mu=\begin{pmatrix}
0.5\\
0.5\\
0.5
\end{pmatrix}, \quad \Sigma=\begin{pmatrix}
1&-0.7&0.3\\
\cdot&1&0.3\\
\cdot&\cdot&1
\end{pmatrix}.
\end{align}
\noindent Each consumer $i$ have three tastes shocks $\nu_i$ generated from independent normal distributions with means equal $(0.5,0.5,0.2)$ and standard deviations all equal $1$. Different draws of consumers are used for simulation and estimation thus leading to sampling error in the estimates.

\begin{table}[ht]
\centering
\caption{Robustness checks for Multisegment Price Heterogeneity Example (Table 4)}\label{tab:approxN}
{\begin{threeparttable}
\begin{tabular}{c||rrrrrr}
  \hline
  & \multicolumn{5}{c}{Number of discretization points} \\
Brand & $1000$ & $2000$ & $5000$ & $10,000$ & $20,000$ & $50,000$ \\ 
  \hline
A & 0 \footnotesize{(ref.)} & 0 \footnotesize{(ref.)} & 0 \footnotesize{(ref.)} & 0 \footnotesize{(ref.)} & 0 \footnotesize{(ref.)} & 0 \footnotesize{(ref.)} \\ 
  B & 0.070 & 0.038 & 0.018 & 0.008 & 0.004 & 0.002 \\ 
  C & 0.040 & 0.023 & 0.010 & 0.004 & 0.002 & 0.001 \\ 
  D & 0.038 & 0.022 & 0.009 & 0.003 & 0.002 & 0.001 \\ 
  E & 0.040 & 0.023 & 0.010 & 0.004 & 0.002 & 0.001 \\ 
  F & 0.056 & 0.033 & 0.013 & 0.005 & 0.003 & 0.001 \\ 
  G & 0.049 & 0.028 & 0.011 & 0.004 & 0.002 & 0.001 \\ 
  H & 0.049 & 0.027 & 0.011 & 0.004 & 0.002 & 0.001 \\ 
   \hline
\end{tabular}
\begin{tablenotes}
\linespread{1}\footnotesize{
\item Estimates of the lower and the upper bounds are LP solution computed with Gurobi 8.1. $50$ Monte-Carlo estimations were made with different draws of $\varepsilon^{a}$. The average difference between the lower and the upper bound is reported. Brand A is the reference and the systematic utility is normalized to $0$.}
\end{tablenotes}
\end{threeparttable}}

\end{table}

\end{document}